\documentclass[11pt,letterpaper]{article}

\usepackage[letterpaper,margin=1in]{geometry}

\usepackage{microtype}
\usepackage{amsmath,amssymb,amsthm,mathtools}
\usepackage{booktabs}
\usepackage{enumitem}
\usepackage{xcolor}
\usepackage[colorlinks=true,linkcolor=blue!50!black,citecolor=green!40!black,urlcolor=blue!50!black]{hyperref}
\usepackage[capitalise,nameinlink]{cleveref}
\usepackage{longtable}
\usepackage{multirow}
\setlist{itemsep=2pt,topsep=4pt}
\setlist[itemize]{label=$\bullet$}
\usepackage{colortbl}
\definecolor{LightCyan}{rgb}{0.8, 0.9, 1}
\usepackage{tikz}
\usetikzlibrary{arrows.meta,positioning,calc,decorations.pathreplacing}

\newtheorem{theorem}{Theorem}[section]
\newtheorem{lemma}[theorem]{Lemma}
\newtheorem{corollary}[theorem]{Corollary}

\theoremstyle{definition}
\newtheorem{definition}[theorem]{Definition}

\newtheorem{procedure}[theorem]{Procedure}
\theoremstyle{remark}

\crefname{construction}{Construction}{Constructions}
\crefname{procedure}{Procedure}{Procedures}
\crefname{fact}{Fact}{Facts}
\newtheorem*{example}{Example}

\newcommand{\F}{\mathbb{F}}

\newcommand{\dL}{d_{\mathrm{L}}}
\newcommand{\val}{\mathrm{val}}
\newcommand{\bin}{\mathrm{bin}}
\newcommand{\Good}{\mathrm{Good}}
\newcommand{\Bad}{\mathrm{Bad}}
\newcommand{\Enc}{\mathrm{Enc}}
\newcommand{\Dec}{\mathrm{Dec}}
\newcommand{\Tr}{\mathrm{Tr}}

\newcommand{\poly}{\mathrm{poly}}
\newcommand{\polylog}{\mathrm{polylog}}

\newcommand{\Tag}{\mathrm{Tag}}

\makeatletter
\renewcommand{\@fnsymbol}[1]{\ensuremath{\ifcase#1\or *\or \dagger\or \ddagger\else \star\fi}}
\makeatother

\title{Seed Recycling: Explicit Codes for Mixed Synchronization Errors\\ Matching the Leading Term of the Gilbert--Varshamov Bound}

\author{
  \begin{tabular}{c}
    Quan D.\ Bui\\[1pt]
    \small VinUniversity, Hanoi, Vietnam\\
    \small\texttt{quan.bd@vinuni.edu.vn}
  \end{tabular}
  \and
  \begin{tabular}{c}
    Manabu Hagiwara\\[1pt]
    \small Chiba University, Chiba, Japan\\
    \small\texttt{hagiwara@math.s.chiba-u.ac.jp}
  \end{tabular}
  \and
  \begin{tabular}{c}
    Van Khu Vu\\[1pt]
    \small VinUniversity, Hanoi, Vietnam\\
    \small\texttt{khu.vv@vinuni.edu.vn}
  \end{tabular}
  \and
  \begin{tabular}{c}
    Shuche Wang\\[1pt]
    \small National University of Singapore, Singapore\\
    \small\texttt{shuche.wang@u.nus.edu}
  \end{tabular}
}
\date{}

\begin{document}

\maketitle
\thispagestyle{empty}

\begin{abstract}
Correcting synchronization errors is one of the oldest problems in coding theory. In the mid-1960s, Levenshtein, and independently Ullman, showed that a single deletion can be corrected with $\log n+O(1)$ bits of redundancy, which is asymptotically optimal. For any fixed $t$, a Gilbert--Varshamov (GV) type argument shows that codes correcting $t$ deletions exist with redundancy $2t\log n+O_t(1)$. Achieving this classical GV-type bound explicitly has proved remarkably difficult. Explicit codes with redundancy $O_t(\log n)$ were first obtained by Brakensiek, Guruswami, and Zbarsky in 2016, half a century after Levenshtein, and Guruswami and H\aa stad matched the GV-type bound for two deletions only in 2021. For every fixed
$t\ge 3$, the best explicit constructions still require roughly twice the GV-type redundancy: $(4t-1)\log n+o(\log n)$ for $t$ deletions and
$4t\log n+O_t(\log\log n)$ for $t$ edits (insertions, deletions, and
substitutions).

We close this factor-of-two gap for every fixed $t$, even in the Damerau--Levenshtein model, which also allows adjacent transpositions. We construct explicit binary codes with deterministic polynomial-time encoding and decoding that correct any combination of at most $t$ such errors using $2t\log n+2t\log\log n+O_t(\log\log\log n)$ bits of redundancy. A single construction thus achieves the leading term of the classical GV-type bound simultaneously for deletions, edits, and Damerau--Levenshtein errors,
extending the Guruswami--H\aa stad result from two deletions to every fixed $t$ and to substantially richer error models.

Our main technical contribution is \emph{seed recycling}, which nearly halves the overhead of hashing-based separation. Standard approaches, including syndrome compression and distributed graph coloring, spend roughly $2\log L$ bits to distinguish a message from $L$ confusable competitors: $\log L$ bits for a hash value and $\log L$ for the hash seed. Over a sufficiently large field, most seeds separate the message from all its competitors, so the
encoder's choice among valid seeds can itself carry message bits, which the decoder recovers from the seed's rank after reconstructing the message.
Combined with a linear multiply-and-project hash, this yields separating tags of net redundancy $\log L+O(1)$, a mechanism that may be of independent
interest.
\end{abstract}

\newpage
\tableofcontents
\thispagestyle{empty}
\newpage
\setcounter{page}{1}

\section{Introduction}
\label{sec:intro}

Synchronization errors, such as insertions and deletions of symbols, destroy the alignment between transmitted and received strings and are much harder to correct than standard substitutions. They arise in document and file synchronization~\cite{CPSV00}, in DNA-based data storage~\cite{CGK12}, where insertions, deletions, substitutions, and adjacent transpositions all occur, and in racetrack memories~\cite{PHT08,CKVVY18}. Codes for such errors were introduced in 1966 by Levenshtein~\cite{Lev66} and Ullman~\cite{ullman1966near} independently. Although this line of research dates back more than sixty years, it has only gained significant momentum over the past decade~\cite{Mit09,HS21survey,BGZ18,GH21}. While much work in theoretical computer science addresses a constant fraction of errors~\cite{SZ99,HS17,GHS20,GHL21}, this paper studies the classical regime of a fixed constant number $t$ of errors as the block length $n\to\infty$, where the central problem is to minimize the required redundancy~\cite{BGZ18,SB21,GH21}.

We consider binary codes $\mathcal C\subseteq\{0,1\}^n$ correcting up to $t$ errors in three standard models: $t$ deletions, which is equivalent to correcting any combination of $t$ insertions and deletions~\cite{Lev66}; $t$ edits, consisting of insertions, deletions, and substitutions; and $t$ Damerau--Levenshtein (DL) errors, which add adjacent transpositions $u_iu_{i+1}\mapsto u_{i+1}u_i$ to the edit operations~\cite{Dam64}. A code corrects $t$ errors if and only if the error balls of radius $t$ around distinct codewords are disjoint. Two inputs are confusable if their radius-$t$ balls intersect. In each of the three models, every input in $\{0,1\}^n$ is confusable with at most $L=O_t(n^{2t})$ other inputs (its \emph{competitors}). Consequently, a greedy Gilbert--Varshamov (GV) type argument guarantees the existence of codes with redundancy $n-\log|\mathcal C|\le\log(L+1)=2t\log n+O_t(1)$, while a sphere-packing argument gives the lower bound $t\log n-O_t(1)$~\cite{Lev66,Lev02}. Recently, Alon, Bourla, Graham, He, and Kravitz~\cite{ABGHK24} improved the existential upper bound for deletions by a $\log\log n$ term. The leading term of this existential upper bound stood for six decades. While we prepare this manuscript, En~Gad~\cite{EnGad26} showed that for $t\ge 2$,  $(2t-1)\log n+O_t(\log\log n)$ redundancy suffices to correct $t$ deletions and also $t$ bounded substring-edits.

The GV-type argument is non-constructive, and matching its leading term $2t\log n$ with explicit codes that have deterministic polynomial-time encoding and decoding has been a long-standing challenge. For $t=1$, the Varshamov--Tenengolts (VT) code and Levenshtein's variant~\cite{VT65,Lev66} achieve $\log n+O(1)$ redundancy for one deletion or one edit, and Gabrys, Guruswami, Ribeiro, and Wu~\cite{GGRW23} showed that codes correcting one deletion or one adjacent transposition with $\log n+O(\log\log n)$ redundancy exist. For $t=2$ deletions, Guruswami and H\aa stad~\cite{GH21} achieved $4\log n+O(\log\log n)$ redundancy, matching the leading term of the GV-type bound. For $t\ge 3$ deletions, and for $t\ge 2$ edits or DL errors, no explicit construction matching the $2t\log n$ bound was previously known, and existing general methods paid roughly twice the GV-type bound.

\subsection{Our results}
\label{subsec:intro-results}

In this paper, we overcome this barrier for every fixed $t\ge 2$. We construct explicit binary codes with deterministic polynomial-time encoding and decoding that match the leading term $2t\log n$ of the GV-type bound for any combination of insertions, deletions, substitutions, and adjacent transpositions.

\begin{theorem}[Main theorem]
\label{thm:main}
For every fixed $t\ge 2$ and all sufficiently large $n$, there are explicit binary codes $\mathcal C\subseteq\{0,1\}^n$ correcting any combination of at most $t$ insertions, deletions, substitutions, and adjacent transpositions, with deterministic encoders and decoders, in two variants:
\begin{enumerate}[label=(\alph*),leftmargin=*]
  \item (bootstrapped) after a deterministic precomputation of $\operatorname{polylog}(n)$ time, the code has redundancy $n-\log|\mathcal C|\le 2t\log n+2t\log\log n+O_t(\log\log\log n)$;
  \item (greedy) after a deterministic precomputation of $n^{6t+2+o(1)}$ time, the code has redundancy $n-\log|\mathcal C|\le 2t\log n+2t\log\log n+O_t(1)$.
\end{enumerate}
In both variants, encoding and full-message decoding take $O_t(n^{2t+2}\log^2 n)$ bit operations, and the payload, consisting of the first $n-O_t(\log n)$ bits of the codeword, is recovered from any channel output in $O_t(n^{t+1}\log n)$ bit operations.
\end{theorem}

\begin{table}[t]
\centering
\caption{Selected redundancy upper bounds for explicit binary codes, for fixed $t$ and $n\to\infty$. The last column uses variant~(a) of Theorem~\ref{thm:main}; every row follows from the same DL construction. Full comparisons are given in Tables~\ref{tab:deletions}--\ref{tab:dl} in Appendix~\ref{app:notation}.}
\label{tab:intro-comparison}\label{tab:sota}
\resizebox{\textwidth}{!}{%
\renewcommand{\arraystretch}{1.35}%
\setlength{\aboverulesep}{0pt}%
\setlength{\belowrulesep}{0pt}%
\begin{tabular}{ll>{\columncolor{LightCyan}}l}
\toprule
\textbf{Errors} & \textbf{Previous constructions} & \multicolumn{1}{l}{\textbf{This work (Theorem~\ref{thm:main}(a))}} \\
\midrule
Two deletions & $4\log n+10\log\log n+O(1)$~\cite{GH21} & $\mathbf{4}\log n+\mathbf{4}\log\log n+O(\log\log\log n)$ \\
\hline
$t$ deletions ($t\ge 3$) & \begin{tabular}[c]{@{}l@{}}$4t\log n+o(\log n)$~\cite{SGB20opt}\\[-2pt] $(4t-1)\log n+o(\log n)$~\cite{SPCH21,SPCH22}\end{tabular} & $\boldsymbol{2t}\log n+2t\log\log n+O_t(\log\log\log n)$ \\
\hline
Two edits & \begin{tabular}[c]{@{}l@{}}$6\log n+O(\log\log n)$~\cite{SG24}\\[-2pt] $6\log n+8$~\cite{PZ24}\end{tabular} & $\mathbf{4}\log n+4\log\log n+O(\log\log\log n)$ \\
\hline
$t$ edits ($t\ge 3$) & \begin{tabular}[c]{@{}l@{}}$4t\log n+o(\log n)$~\cite{SGB20opt,SGB20syndrome} \\[-2pt] $4t\log n+O_t(\log\log n)$~\cite{LGF25}\end{tabular} & $\boldsymbol{2t}\log n+2t\log\log n+O_t(\log\log\log n)$ \\
\hline
$t$ DL errors ($t\ge 2$) & $8t\log n+O_t(\log\log n)$~\cite{LGF25} (via $2t$-edits codes) & $\boldsymbol{2t}\log n+2t\log\log n+O_t(\log\log\log n)$ \\
\bottomrule
\end{tabular}%
}
\end{table}

Table~\ref{tab:sota} summarizes how Theorem~\ref{thm:main}(a) compares with the best previous explicit constructions; a fuller historical overview is deferred to Appendix~\ref{app:notation}. For $t$ deletions, after early explicit codes with $O(t^2\log t\log n)$ and $O(t\log n)$ redundancy~\cite{BGZ18,CJLW18,Hae19} and two-deletion codes with $8\log n$ and $7\log n$ redundancy~\cite{GS19,SRB20}, Guruswami and H\aa stad~\cite{GH21} achieved $4\log n+O(\log\log n)$ for $t=2$. This is the first known explicit binary code that its redundancy match the leading term of the GV-type upper bound. For $t\ge 3$ deletions, Sima and Bruck~\cite{SB21} obtained $8t\log n+o(\log n)$, syndrome compression~\cite{SGB20syndrome,SGB20opt} gave $4t\log n+o(\log n)$, and Song, Polyanskii, Cai, and He~\cite{SPCH21,SPCH22} reduced this to $(4t-1)\log n+o(\log n)$ by precoding with a single-deletion syndrome. For two edits, Sun and Ge~\cite{SG24} and Pi and Zhang~\cite{PZ24} constructed explicit codes with redundancy $6\log n+O(\log\log n)$ and $6\log n+8$, respectively. 

For general $t$ edits, Li, Gabrys, and Farnoud~\cite{LGF25} obtained $4t\log n+O(\log\log n)$ redundancy via Linial's coloring, and combining the framework of \cite{LGF25} with Levenshtein’s single-edit syndrome yields $(4t-1)\log n+O_t(\log\log n)$ (see Appendix~\ref{app:lgf-consequences}). Con and Li~\cite{con2025oblivious} recently studied oblivious deletion codes, where the deletion pattern is independent of the encoder's private randomness, and also gave a randomized construction that, with high probability over the random choice of the code, corrects $t$ adversarial deletions with redundancy $(2t+1)\log n+O_t(\log\log n)$. In contrast, our construction is fully deterministic, including its precomputation, and achieves $2t\log n+O_t(\log\log n)$ redundancy for any combination of $t$ edits and adjacent transpositions. 

For $t$ DL errors, to the best of our knowledge, no dedicated explicit code was known for general $t\ge 2$, with prior works limited to special combinations such as one deletion and $\ell$ adjacent transpositions~\cite{GYM18,WVT23}. Since an adjacent transposition can be simulated by two substitutions, any code correcting $2t$ edits corrects $t$ DL errors, but this naive relaxation gives $8t\log n+O_t(\log\log n)$ bits of redundancy. We observe that combining the $2$-substring edit syndrome of~\cite[Theorem~6]{LGF25} with our DL boundary projection lemma (Lemma~\ref{lem:dl-projection}) would yield an explicit code with $4t\log n+O_t(\log\log n)$ redundancy (see Appendix~\ref{app:lgf-consequences}), which our main construction further improves to $2t\log n+2t\log\log n+O_t(\log\log\log n)$.

Related synchronization channels such as $q$-ary deletions~\cite{Ten84,LN16,SGB20optq,SC23,LTX24,NCS24,STV25}, permutation codes~\cite{Lev92,JMSB09,FSM13,GYFSBD16,WNCV25,WBNVW26}, and $\ell$-symbol read channels~\cite{CB11,CV20,XC26,YG26} are not treated here. 

\subsection{Technique Overview}
Our construction has three components:
\begin{enumerate}[leftmargin=*]
  \item \emph{ Seed recycling and breaking the factor of two (Section~\ref{sec:recycling}).} To understand the source of our improvement, we first identify why previous general methods lose a factor of two in redundancy. Consider the basic task of separating each input message $x\in\{0,1\}^m$ from a set $\Gamma(x)$ of at most $L$ confusable competitors. In both syndrome compression~\cite{SGB20syndrome} and distributed graph coloring~\cite{LGF25}, the standard approach selects a hash function $h_a$ from a hash family that avoids collisions with all competitors in $\Gamma(x)$, and appends the tag $(a, h_a(x))$. Because the tag must describe both the seed $a$ (an integer modulus $p$ in syndrome compression) and the hash value $h_a(x)$ (the residue modulo $p$), and each component costs about $\log L$ bits, the resulting tag requires about $2\log L$ bits of redundancy.

  \begin{table}[ht]
\centering
\caption{Net cost of a tag separating an input from at most $L$ competitors before protecting the tag (for inputs of length $n$ with $L=n^{O(1)}$).}
\label{tab:factor-two}
\resizebox{\textwidth}{!}{%
\begin{tabular}{lll}
\hline
\toprule
\textbf{Method} & \textbf{Tag components} & \textbf{Net cost} \\
\midrule
Syndrome compression~\cite{SGB20syndrome} & modulus $p$ ($\log L+o(\log n)$) $+$ residue ($\log L+o(\log n)$) & $2\log L+o(\log n)$ \\
Linial recoloring~\cite{LGF25,Lin92} & one color from a palette of size $L^2\operatorname{poly}(\log n)$ & $2\log L+O(\log\log n)$ \\
\rowcolor{LightCyan} \textbf{Seed recycling} (Theorem~\ref{thm:recycling}) & $2b$-bit seeds (carrying $2b-2$ message bits) $+$ $\lceil\log(2L)\rceil$-bit hash & $<\log L+4$ \\
\bottomrule
\end{tabular}%
}
\end{table}

  Our main tool, \emph{seed recycling}, eliminates this factor-of-two penalty by making the hash seed itself nearly free. If we draw the seed $a$ from a sufficiently large finite field $\mathbb F_{2^b}$, more than half of all $2^b$ seeds separate $x$ from all $L$ of its competitors. Instead of picking an arbitrary good seed, the encoder uses $b-1$ additional message bits as an index to select which of the first $2^{b-1}$ good seeds to use in a canonical order; once the decoder recovers $x$, it reconstructs the same ordered list of good seeds and reads those $b-1$ message bits back from the seed's rank. Thus, a $b$-bit seed carries $b-1$ message bits and contributes only $1$ bit to the net redundancy. See Table~\ref{tab:factor-two} for the detailed comparison.

  In the first stage, we view $x$ as a polynomial $P_x$ of degree at most $k-1$ over $\mathbb F_{2^b}$, where $k=\lceil m/b\rceil$. Since each competitor can collide with $x$ on at most $k-1$ roots, choosing $2^b\ge 2Lk$ ensures that more than half of the seeds are good. However, the remaining $b$-bit hash value $P_x(a)$ still depends on the number of blocks $k$. A second hashing step removes this dependence: for a good seed $a$, every difference $P_x(a)-P_{x'}(a)$ is nonzero, so multiplying $P_x(a)$ by a second seed $\lambda\in\mathbb F_{2^b}$ and projecting onto the first $r=\lceil\log(2L)\rceil$ bits fails for \emph{exactly} a $2^{-r}$ fraction of multipliers $\lambda$ for each fixed competitor $x'$, independent of $k$. Recycling $\lambda$ in the same way yields a tag $\Tag(x,\sigma,\tau)$ of net cost $r+2<\log L+4$ bits (Theorem~\ref{thm:recycling}; see Figure~\ref{fig:overview}). 
  
 % From a graph-theoretic viewpoint, the set of tags assigned to $x$ as $(\sigma,\tau)$ varies forms a local multicoloring of the conflict graph with ratio $2^{r+2}<16L$. This viewpoint connects our construction to Kuhn’s polynomial-based local multicoloring algorithm~\cite{kuhn2009local}. Our construction uses polynomial evaluation in the first stage and field multiplication followed by a fixed binary projection in the second, in place of further polynomial evaluation. For coding, we use the indices $(\sigma,\tau)$ of the selected seeds to carry additional message bits; these indices can be recovered from the tag once $x$ is known.

  \item \emph{Protecting the tag on splittable channels (Section~\ref{sec:tags-to-codes}).} To protect the tag against channel errors, we encode it with a short tag code $A:\{0,1\}^h\to\{0,1\}^N$ and append it to the payload $x\in\{0,1\}^m$ to form the codeword $x\Vert A(\Tag(x,\sigma,\tau))$. Although synchronization errors shift the boundary between $x$ and the suffix in the received word, we show (Lemma~\ref{lem:tag-to-code} and Theorem~\ref{thm:efficient-channel-code}) that this construction succeeds whenever the channel is \emph{splittable}, meaning that confusable words $x\Vert s$ and $x'\Vert s'$ must have suffixes $s$ and $s'$ related by a symmetric relation $\mathcal R_N$ that $A$ separates, and whenever suffix neighborhoods, payload competitors, and channel preimages can be listed in polynomial time. Splittability forces confusable codewords to share the exact same tag $T$ and suffix $s=A(T)$, so $x$ only needs to be separated from competitors $x'$ that confuse with $x$ when followed by a common suffix $s$. In our DL application, the tag has length $h=O_t(\log n)$, and protecting it adds only $2t\log\log n+O_t(1)$ bits of overhead with a greedy tag code, or $2t\log\log n+O_t(\log\log\log n)$ bits when we bootstrap our own code family at length $O_t(\log n)$ to reduce the precomputation time to $\operatorname{polylog}(n)$.

  \item \emph{Boundary tracking for DL errors (Section~\ref{sec:dl-codes}).} For the Damerau--Levenshtein metric, an adjacent transposition across the boundary between $x$ and $s$ modifies both parts in a single operation. By tracking a boundary marker through a shortest sequence of operations and bounding the suffix distortion together with the displacement of the marker, we prove the suffix projection bound $d_{\mathrm{DL}}(s,s')\le d_{\mathrm{DL}}(x\Vert s,x'\Vert s')$ for all $x,x'\in\{0,1\}^m$ and $s,s'\in\{0,1\}^N$ (Lemma~\ref{lem:dl-projection}). Since string reversal preserves $d_{\mathrm{DL}}$ and exchanges prefixes with suffixes, the same bound holds for the prefixes: $d_{\mathrm{DL}}(x,x')\le d_{\mathrm{DL}}(x\Vert s,x'\Vert s')$. In particular, $d_{\mathrm{DL}}(x\Vert s,x'\Vert s)=d_{\mathrm{DL}}(x,x')$, which identifies the exact payload competitor set of $x$ as the same-length DL ball $(B_{2t}(x)\cap\{0,1\}^m)\setminus\{x\}$ of size at most $(2t+1)10^{2t}n^{2t}$ (Corollary~\ref{cor:dl-competitors}).
\end{enumerate}

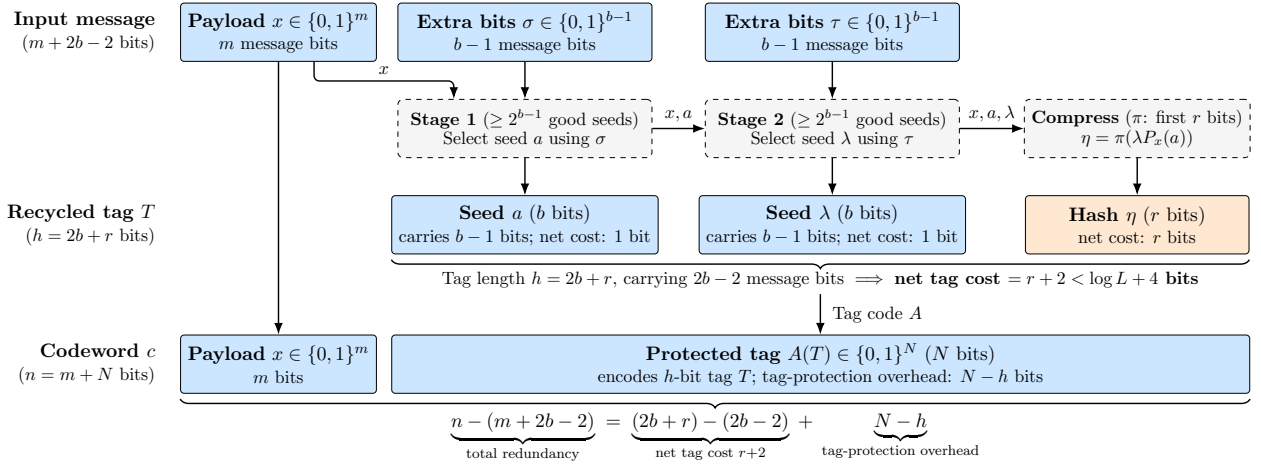
\begin{figure}[t]
\centering
\small
\resizebox{\textwidth}{!}{%
\begin{tikzpicture}[
    font=\normalsize,
    >=Latex,
    box/.style={draw, rounded corners=2pt, minimum height=1.15cm, align=center, inner sep=4pt},
    msgbox/.style={box, fill=LightCyan},
    redbox/.style={box, fill=orange!18},
    procbox/.style={draw, dashed, rounded corners=3pt, fill=gray!8, minimum height=1.15cm, inner sep=4pt, align=center, font=\small},
    bracebot/.style={decorate, decoration={brace, mirror, amplitude=5pt, raise=3pt}, thick}
]

% ==================== ROW 1: INPUT MESSAGE (y = 0) ====================
\node[msgbox, minimum width=3.8cm] (in_x) at (0, 0) {\textbf{Payload} $x \in \{0,1\}^m$\\[-1pt]{\small $m$ message bits}};
\node[msgbox, minimum width=5.0cm] (in_sig) at (4.85, 0) {\textbf{Extra bits} $\sigma \in \{0,1\}^{b-1}$\\[-1pt]{\small $b-1$ message bits}};
\node[msgbox, minimum width=5.0cm] (in_tau) at (10.90, 0) {\textbf{Extra bits} $\tau \in \{0,1\}^{b-1}$\\[-1pt]{\small $b-1$ message bits}};

\node[anchor=east, align=right] at (-2.30, 0) {\textbf{Input message}\\[-1pt]{\small ($m+2b-2$ bits)}};

% ==================== ROW 2: TWO-STAGE SELECTION (y = -1.90) ====================
\node[procbox, minimum width=5.0cm] (good1) at (4.85, -1.90) {\textbf{Stage 1} ($\ge 2^{b-1}$ good seeds)\\[-1pt] Select seed $a$ using $\sigma$};
\node[procbox, minimum width=5.0cm] (good2) at (10.90, -1.90) {\textbf{Stage 2} ($\ge 2^{b-1}$ good seeds)\\[-1pt] Select seed $\lambda$ using $\tau$};
\node[procbox, minimum width=4.4cm] (hash) at (16.90, -1.90) {\textbf{Compress} ($\pi$: first $r$ bits)\\[-1pt] $\eta = \pi(\lambda P_x(a))$};

% Arrows into and across Row 2
\draw[->, thick] (in_sig.south) -- (good1.north);
\draw[->, thick] (in_tau.south) -- (good2.north);
\draw[->, thick, rounded corners=3pt] ([xshift=0.7cm]in_x.south) -- ++(0,-0.38) -| node[pos=0.25, above, font=\small] {$x$} ([xshift=-1.4cm]good1.north);
\draw[->, thick] (good1.east) -- node[above, font=\small] {$x,a$} (good2.west);
\draw[->, thick] (good2.east) -- node[above, font=\small] {$x,a,\lambda$} (hash.west);

% ==================== ROW 3: RECYCLED TAG (y = -3.80) ====================
\node[msgbox, minimum width=5.0cm] (tag_a) at (4.85, -3.80) {\textbf{Seed $a$} ($b$ bits)\\[-1pt]{\small carries $b-1$ bits; net cost: $1$ bit}};
\node[msgbox, minimum width=5.0cm] (tag_l) at (10.90, -3.80) {\textbf{Seed $\lambda$} ($b$ bits)\\[-1pt]{\small carries $b-1$ bits; net cost: $1$ bit}};
\node[redbox, minimum width=4.4cm] (tag_eta) at (16.90, -3.80) {\textbf{Hash $\eta$} ($r$ bits)\\[-1pt]{\small net cost: $r$ bits}};

\node[anchor=east, align=right] at (-2.30, -3.80) {\textbf{Recycled tag $T$}\\[-1pt]{\small ($h = 2b+r$ bits)}};

\draw[->, thick] (good1.south) -- (tag_a.north);
\draw[->, thick] (good2.south) -- (tag_l.north);
\draw[->, thick] (hash.south) -- (tag_eta.north);

\draw[bracebot] (tag_a.south west) -- (tag_eta.south east) node[midway, below=6pt, font=\small] (tag_brace) {Tag length $h = 2b+r$, carrying $2b-2$ message bits $\implies$ \textbf{net tag cost $= r+2 < \log L + 4$ bits}};

% ==================== ROW 4: FINAL CODEWORD (y = -6.55) ====================
\node[msgbox, minimum width=3.8cm] (cw_x) at (0, -6.55) {\textbf{Payload} $x \in \{0,1\}^m$\\[-1pt]{\small $m$ bits}};
\path let \p1=(tag_a.west), \p2=(tag_eta.east) in node[msgbox, anchor=west, minimum width=\x2-\x1] (cw_A) at (\x1, -6.55) {\textbf{Protected tag} $A(T) \in \{0,1\}^N$ ($N$ bits)\\[-1pt]{\small encodes $h$-bit tag $T$; tag-protection overhead: $N-h$ bits}};

\node[anchor=east, align=right] at (-2.30, -6.55) {\textbf{Codeword $c$}\\[-1pt]{\small ($n = m+N$ bits)}};

\draw[->, thick] (in_x.south) -- (cw_x.north);
\draw[->, thick] (tag_brace.south) -- node[right=3pt, font=\small] {Tag code $A$} (tag_brace.south |- cw_A.north);

% ==================== BOTTOM ACCOUNTING FORMULA ====================
\draw[bracebot] (cw_x.south west) -- (cw_A.south east) node[midway, below=7pt, font=\normalsize] {$\underbrace{n - (m+2b-2)}_{\text{total redundancy}} \;=\; \underbrace{(2b+r) - (2b-2)}_{\text{net tag cost } r+2} \;+\; \underbrace{N - h}_{\text{tag-protection overhead}}$};

\end{tikzpicture}%
}
\caption{Overview of two-stage seed recycling and codeword construction (Sections~\ref{sec:recycling} and~\ref{sec:tags-to-codes}). Blue blocks carry message bits, while the orange block is pure hash redundancy. Both good-seed sets and the compressed hash value depend on the payload $x$ and its competitor set $\Gamma(x)$: each $b$-bit seed ($a$ and $\lambda$) is selected from at least $2^{b-1}$ valid seeds using a $(b-1)$-bit message block ($\sigma$ and $\tau$), contributing only $1$ bit to the net redundancy.}
\label{fig:overview}
\end{figure}

\subsection{Outline and notation}
In Section~\ref{sec:recycling}, we present our seed recycling construction (Theorem~\ref{thm:recycling}). In Section~\ref{sec:tags-to-codes}, we introduce splittable channels and prove our general code construction theorem (Theorem~\ref{thm:efficient-channel-code}). In Section~\ref{sec:dl-codes}, we establish the boundary projection property and prove our main result for Damerau--Levenshtein errors (Theorem~\ref{thm:dl-codes}, which proves Theorem~\ref{thm:main}), followed by open problems in Section~\ref{sec:discussion}. Appendix~\ref{app:notation} provides the full historical comparison tables and a summary of notation, Appendix~\ref{app:section23} gives the formal specification of the generic decoder, Appendix~\ref{sec:algorithms} gives the deterministic finite-field and enumeration algorithms, and Appendix~\ref{sec:bootstrap} proves the bootstrapping lemma.

Throughout the paper, all logarithms are to base $2$. For integers $a\le b$, we write $[a,b]=\{a,\ldots,b\}$ and $[K]=[1,K]$. For a string $u=u_1\cdots u_\ell\in\{0,1\}^\ell$, we write $|u|=\ell$, $u_{[i,j]}=u_i\cdots u_j$, $\val(u)=\sum_{i=1}^\ell u_i 2^{\ell-i}\in[0,2^\ell-1]$, and $\bin_\ell:[0,2^\ell-1]\to\{0,1\}^\ell$ for the inverse of $\val$. Concatenation is written as $uv$ or $u\Vert v$, and $\ker$ denotes the kernel of an $\mathbb F_2$-linear map. The number of errors $t\ge 1$ is a fixed constant, and $O_t(\cdot)$ hides factors depending only on $t$.
% =====================================================================

% =====================================================================
\section{Seed recycling}
\label{sec:recycling}

In our basic setup, we have an unknown string $x\in\{0,1\}^{m}$ which is corrupted by a channel to obtain a received string $y$. The goal is to reconstruct $x$ from $y$. In general, $y$ alone is not sufficient to identify $x$ uniquely, and we wish to provide a short ``tag'' (or sketch) of $x$ that enables the decoder to recover $x$. Since the encoder does not know which output $y$ will be received, the tag must distinguish $x$ from every other string $x'$ that can produce a common output with $x$. We refer to such strings as the \emph{competitors} of $x$.
Let
\begin{equation*}
  \Gamma(x)\subseteq\{0,1\}^{m}\setminus\{x\}
  \qquad\text{with}\qquad |\Gamma(x)|\le L
\end{equation*}
be a set containing all competitors of $x$, where $m,L\ge1$ are integers. For now, we assume that the tag is known to the decoder without errors; in
Section~\ref{sec:tags-to-codes}, we will see how to encode and append the tag to $x$ itself. The construction in this section uses only $x$ and $\Gamma(x)$.

A standard way to construct a separating tag is to choose a hash function that separates $x$ from every string in $\Gamma(x)$, and specify both the seed of the hash function and the hash value (as in the syndrome compression method of~\cite{SGB20syndrome}). This requires paying for both the seed and the hash output. Here we use a simple observation that if the pool of seeds is large enough, then \emph{many} seeds are valid for the same $x$, and it does not matter which valid seed we use. We can therefore order all valid seeds for $x$ and use $c$ additional message bits as an index to select which valid seed to include in the tag. Once the decoder has reconstructed $x$, it can compute the same ordered list of valid seeds, find the position of the chosen seed in this list, and read off the $c$ additional message bits. Thus, if the tag has length $h$ in total, its \emph{net cost} in redundancy is only $h-c$ bits.

We carry out this idea in two steps: we first apply seed recycling to polynomial hashing over $\F_{2^b}$, and then add a second, linear hashing step to compress the hash value itself.

\subsection{Hashing and seed selection}
\label{subsec:first-seed}

Let $Q=2^b$, and fix an $\F_2$-linear bijection $\varphi:\{0,1\}^{b}\to\F_Q$, which orders $\F_Q$ by the integer value $\iota(a)=\val(\varphi^{-1}(a))$. We pad $x\in\{0,1\}^{m}$ with zeroes to length $kb$, where $k=\lceil m/b\rceil$, and divide it into $k$ blocks of $b$ bits. Viewing these blocks as field elements $X_0,\ldots,X_{k-1}\in\F_Q$, we define $P_x(Z)=\sum_{i=0}^{k-1}X_i Z^i$. For any $x'\ne x$, the difference $P_x-P_{x'}$ is a nonzero polynomial of degree at most $k-1$, and thus has at most $k-1$ roots in $\F_Q$.

We call a seed $a\in\F_Q$ \emph{good} for $x$ if it avoids collisions with all competitors of $x$, and set
\begin{equation*}
  \Good_1(x)=
  \bigl\{a\in\F_Q:
  P_x(a)\ne P_{x'}(a)\text{ for every }x'\in\Gamma(x)\bigr\}.
\end{equation*}
Since each of the at most $L$ competitors rules out at most $k-1$ seeds,
choosing $b$ such that $2^b\ge 2L\lceil m/b\rceil$ ensures that more than half of all seeds in $\F_Q$ are good:
\begin{equation}
  |\Good_1(x)|\ge Q-L(k-1)>Q/2.
  \label{eq:first-good-count}
\end{equation}

Because $|\Good_1(x)|>Q/2=2^{b-1}$, the encoder can use $b-1$ additional
message bits to select one of the first $2^{b-1}$ seeds in $\Good_1(x)$, and
the decoder can recover these bits from the rank of $a$ in $\Good_1(x)$ once
$x$ is known. Thus, the $2b$-bit tag $(a,P_x(a))$ has net cost
$2b-(b-1)=b+1$. However, since $2^b\ge 2L\lceil m/b\rceil$ requires
$b\ge\log(2Lk)$, the $b$-bit hash value $P_x(a)$ makes this net cost grow
with the number of blocks $k$. We next shorten this hash value to remove the
dependence on $k$.

\subsection{Compressing the hash value}
\label{subsec:second-seed}

Suppose we have fixed a good seed $a\in\Good_1(x)$. By definition, for every competitor $x'\in\Gamma(x)$, the difference $\delta_{x'}=P_x(a)-P_{x'}(a)$ is a nonzero element of $\mathbb F_Q$. Let $1\le r\le b$, and let $\pi:\mathbb F_Q\to\{0,1\}^r$ be the $\mathbb F_2$-linear projection that outputs the first $r$ bits of the binary representation of a field element, with kernel $\ker\pi=\{z\in\mathbb F_Q:\pi(z)=0^r\}$ of size $2^{b-r}$. We consider the second-stage hash function $z\mapsto\pi(\lambda z)$ for $\lambda\in\mathbb F_Q$. For any nonzero $\delta\in\mathbb F_Q$, the map $\lambda\mapsto\lambda\delta$ is a bijection on $\mathbb F_Q$, so there are exactly $|\ker\pi|=2^{b-r}$ values of $\lambda\in\mathbb F_Q$ for which $\pi(\lambda\delta)=0^r$. In other words, each competitor $x'\in\Gamma(x)$ rules out a $2^{-r}$ fraction of the multipliers $\lambda$, and this fraction is independent of the degree of $P_x$.

Analogously to the first stage, we define the set of good multipliers for
$(x,a)$ by
\begin{equation*}
  \Good_2(x,a)=
  \bigl\{\lambda\in\F_Q:
  \pi(\lambda P_x(a))\ne\pi(\lambda P_{x'}(a))
  \text{ for every }x'\in\Gamma(x)\bigr\}.
\end{equation*}
Taking a union bound over the at most $L$ competitors in $\Gamma(x)$, we
have
\begin{equation}
  |\Good_2(x,a)|\ge Q-L2^{b-r}\ge Q/2,
  \qquad\text{provided that }2^r\ge2L.
  \label{eq:second-good-count}
\end{equation}
Thus, we can use another $b-1$ message bits to select one of the first $Q/2$ elements of $\Good_2(x,a)$. Together, the two seeds $a$ and $\lambda$ take $2b$ bits to specify and carry $2b-2$ additional message bits. Their combined contribution to the \textit{net cost} is only two bits, on top of the $r$ bits needed for the hash value $\pi(\lambda P_x(a))$.

\subsection{The tag and its recovery}
\label{subsec:recycled-tag}

We now state the complete construction. Fix integers $b,r$ satisfying 
\begin{equation}
  2^b\ge2L\lceil m/b\rceil,
  \qquad 2^r\ge2L,
  \qquad 1\le r\le b,
  \label{eq:seed-parameters}
\end{equation}
and let $K=2^{b-1}$. Given the string $x\in\{0,1\}^{m}$ and two blocks of additional message bits $\sigma,\tau\in\{0,1\}^{b-1}$, we set $v=1+\val(\sigma), w=1+\val(\tau)$. We then let $a$ be the $v$-th element of $\Good_1(x)$ and $\lambda$ be the $w$-th element of $\Good_2(x,a)$ under the fixed ordering of $\F_Q$. By eqns.~\eqref{eq:first-good-count} and \eqref{eq:second-good-count}, both sets have size at least $K$, so $a$ and $\lambda$ are well-defined. The tag is given by
\begin{equation}
  \Tag(x,\sigma,\tau)
  =\varphi^{-1}(a)\,\|\,\varphi^{-1}(\lambda)\,\|\,
  \pi\bigl(\lambda P_x(a)\bigr).
  \label{eq:tag-definition}
\end{equation}

\begin{example}
Let $m=6$, $b=3$, and $L=r=2$ (so $k=2$, $K=4$). Take $\F_8=\F_2[U]/(U^3+U+1)$ with $\varphi(u_1u_2u_3)=u_1+u_2U+u_3U^2$, which orders $\F_8$ as $0, U^2, U, U+U^2, 1, 1+U^2, 1+U, 1+U+U^2$. Suppose $x=101100$, $\Gamma(x)=\{000000,101000\}$, and the additional message bits are $\sigma=10$ and $\tau=01$.
\begin{itemize}[nosep,leftmargin=*] \item \emph{Stage 1:} We have $P_x(Z)=(1+U^2)+Z$ and $P_{x'}(Z)\in\{0,1+U^2\}$. The differences $P_x-P_{x'}$ have roots $1+U^2$ and $0$, leaving $\Good_1(x)=(U^2, U, U+U^2, 1, 1+U, 1+U+U^2)$. The bits $\sigma=10$ give rank $v=3$, which selects $a=U+U^2$ and $P_x(a)=1+U$.
\item \emph{Stage 2:} At $a=U+U^2$, the differences $\delta=P_x(a)-P_{x'}(a)$ are $1+U$ and $U+U^2$. Since $\ker\pi=\{0,U^2\}$, the bad multipliers are $\cup_{\delta\in\{1+U,U+U^2\}} \delta^{-1}\ker\pi=\{0, 1+U^2, 1+U+U^2\}$, leaving $\Good_2(x,a)=(U^2, U, U+U^2, 1, 1+U)$. The bits $\tau=01$ give rank $w=2$, selecting $\lambda=U$ and hash value $\pi(\lambda P_x(a))=\pi(U+U^2)=01$.
\end{itemize}
The resulting $8$-bit tag is $\Tag(x,10,01)=011\,\|\,010\,\|\,01$, which carries the $4$ bits $\sigma\tau$ at a net cost of $4$ bits. Conversely, given $x$ and the tag, the decoder recomputes $\Good_1(x)$ and $\Good_2(x,a)$ and reads off the ranks $3$ and $2$ to recover $\sigma=10$ and $\tau=01$.
\end{example}

We summarize the properties of our tag construction in the following theorem.

\begin{theorem}[Seed recycling]
\label{thm:recycling}
For every assignment
$\Gamma(x)\subseteq\{0,1\}^{m}\setminus\{x\}$ with $|\Gamma(x)|\le L$,
and any parameters $b,r$ satisfying \eqref{eq:seed-parameters}, the map
\[
  \Tag:\{0,1\}^{m}\times\{0,1\}^{b-1}\times\{0,1\}^{b-1}
  \longrightarrow\{0,1\}^{2b+r}
\]
defined in \eqref{eq:tag-definition} satisfies the following properties:
\begin{enumerate}[label=(\alph*),leftmargin=*]
\item If $x'\in\Gamma(x)$, then $\Tag(x,\sigma,\tau)\ne\Tag(x',\sigma',\tau')$ for all $\sigma,\tau,\sigma',\tau'\in\{0,1\}^{b-1}$.
\item For each fixed $x$, the map
$(\sigma,\tau)\mapsto\Tag(x,\sigma,\tau)$ is injective, and the pair
$(\sigma,\tau)$ can be uniquely recovered from $x$ and $\Tag(x,\sigma,\tau)$.
\item The tag carries $2b-2$ additional message bits and has net cost $r+2$.
In particular, setting $r=\lceil\log(2L)\rceil$ gives net cost less than
$\log L+4$.
\end{enumerate}
\end{theorem}

\begin{proof}
Suppose that $\Tag(x,\sigma,\tau)=\Tag(x',\sigma',\tau')$ for some
$x'\in\Gamma(x)$. Comparing the first two parts of the tags shows that both
tags must use the same seeds $a$ and $\lambda$. Comparing the last part then
gives $\pi(\lambda P_x(a))=\pi(\lambda P_{x'}(a))$, which contradicts the
fact that $\lambda\in\Good_2(x,a)$. Note that this argument uses only the
fact that the seeds are good for $x$, so we do not require $x\in\Gamma(x')$.

For a fixed $x$, distinct values of $\sigma$ select distinct first seeds
$a\in\Good_1(x)$. When both $x$ and $\sigma$ are fixed, distinct values of
$\tau$ select distinct multipliers $\lambda\in\Good_2(x,a)$. Thus the map
$(\sigma,\tau)\mapsto\Tag(x,\sigma,\tau)$ is injective. Once $x$ and its tag
are known, the decoder can reconstruct $\Good_1(x)$ and $\Good_2(x,a)$ and
recover $\sigma$ and $\tau$ from the ranks of $a$ and $\lambda$. Finally,
the net cost of the tag is $(2b+r)-(2b-2)=r+2$, which for
$r=\lceil\log(2L)\rceil$ is less than $\log(2L)+3=\log L+4$.
\end{proof}

\paragraph{Order of recovery.}
Note that in a tag $T=\Tag(x,\sigma,\tau)=\varphi^{-1}(a)\,\|\,\varphi^{-1}(\lambda)\,\|\,\eta$, the seeds $a$ and $\lambda$ are given explicitly. Thus, to test whether a candidate string $z$ matches $T$, the decoder does not need to know the ranks of $a$ and $\lambda$: it simply checks whether $\pi(\lambda P_z(a))=\eta$. If the decoder has a list of candidate strings consisting of $x$ and a subset of $\Gamma(x)$, then $x$ clearly passes this test, whereas every $z\in\Gamma(x)$ fails because $\lambda\in\Good_2(x,a)$ gives $\pi(\lambda P_z(a))\ne\pi(\lambda P_x(a))=\eta$. Once $x$ is uniquely identified as the passing candidate, the decoder reconstructs $\Good_1(x)$ and $\Good_2(x,a)$ and recovers $(\sigma,\tau)$ from the ranks of $a$ and $\lambda$.

\paragraph{Computing the tag.}
If we choose $b=\lceil\log(2Lm)\rceil$ and $r=\lceil\log(2L)\rceil$, then
\eqref{eq:seed-parameters} holds and the field size satisfies $Q=2^b<4Lm$. Consequently, given $x$ and a list of the elements of $\Gamma(x)$, we can compute $\Tag(x,\sigma,\tau)$ and recover $(\sigma,\tau)$ deterministically in time polynomial in $m$, $L$, and the length of the list. Indeed, one can simply test all $Q$ elements of $\F_Q$ against the list of competitors to form $\Good_1(x)$ and $\Good_2(x,a)$. In Appendix~C (Lemma~C.5), we give a more efficient implementation that computes the bad-seed sets $\Bad_1(x)=\F_Q\setminus\Good_1(x), \Bad_2(x,a)=\F_Q\setminus\Good_2(x,a)$ directly by finding roots of difference polynomials and enumerating scaled copies of $\ker\pi$, followed by selection and ranking using the sorted bad-seed lists.

%\paragraph{A multicoloring interpretation.}
%For each $x\in\{0,1\}^{m}$, let $S(x)=\{\Tag(x,\sigma,\tau):\sigma,\tau\in\{0,1\}^{b-1}\}$. By Theorem~\ref{thm:recycling}, each set $S(x)$ contains $2^{2b-2}$ colors, and $S(x)\cap S(x')=\varnothing$ whenever $x'\in\Gamma(x)$. Thus, the sets $S(x)$ form a multicoloring of the conflict graph in which two strings are adjacent if either is a competitor of the other. With $\{0,1\}^{2b+r}$ as the palette and $r=\lceil\log(2L)\rceil$, the ratio of the palette size to the number of colors per vertex is $2^{2b+r}/2^{2b-2}=2^{r+2}<16L$. Kuhn's local multicoloring construction~\cite[Section~4.3]{kuhn2009local} also retains separating polynomial evaluations and composes evaluation maps (achieving a ratio below $16L$ in two stages when $\Gamma$ is symmetric and $L\ge m$, by his Lemma~4.3), whereas our second stage uses multiplication in the same field followed by binary projection, giving a two-stage construction for arbitrary $m$ and $L$. In our coding application, the pair $(\sigma,\tau)$ selects a color from $S(x)$, so the net tag cost is the logarithm of the ratio above, namely $r+2$ bits.

Furthermore, to obtain an error-correcting code, we must also protect the tag against channel errors. The next section gives a sufficient splitting condition under which this can be achieved by encoding the tag with a short error-correcting code and appending the result to the payload.

\section{From tags to error-correcting codes}
\label{sec:tags-to-codes}

In Section~\ref{sec:recycling}, we showed how to construct a short tag that distinguishes a string from its competitors while carrying additional message bits, under the assumption that the tag is known to the decoder without errors. To obtain an actual error-correcting code, we encode the tag itself by a short error-correcting code and append the encoded tag to the message. Since the tag is much shorter than the message, protecting it adds only a small number of extra redundant bits.

\subsection{Protecting the tag}
\label{subsec:protecting-tag}

Let $h=2b+r$ be the length of the tag from Section~\ref{sec:recycling},
and let $A:\{0,1\}^{h}\to\{0,1\}^{N}$ be an encoding that maps the $h$-bit
tag to a protected suffix of length $N$. Given a payload
$x\in\{0,1\}^{m}$ and additional message bits
$\sigma,\tau\in\{0,1\}^{b-1}$, we form a codeword of length $n=m+N$ by
appending the encoded tag to $x$:
\begin{equation}
  \Enc(x,\sigma,\tau)
  =x\Vert A\bigl(\Tag(x,\sigma,\tau)\bigr),
  \qquad n=m+N.
  \label{eq:protected-tag-encoder}
\end{equation}

When channel errors can change the length of a string (such as deletions or insertions), the boundary between the payload $x$ and the suffix $s$ is no longer known in the received word $y$. For example, if we delete the first bit of the suffix $s$ in $101\Vert 010\rightarrow 10110$, the last three bits of the received word are $110$, which mixes a bit from $x$ with bits from $s$. However, for any two transmitted words $x\Vert s$ and $x'\Vert s'$ of length $n=m+N$, the boundary between the $m$-bit prefix and the $N$-bit suffix is fixed before transmission. Thus, if the channel has the property that any two confusable words $x\Vert s$ and $x'\Vert s'$ must have suffixes $s$ and $s'$ that are close to each other, we can design the map $A$ to keep distinct tags far apart.

Formally, a channel $\mathcal E=(\mathcal E_n)_{n\ge1}$ assigns to each
input $u\in\{0,1\}^{n}$ a set $\mathcal E_n(u)\subseteq\{0,1\}^{*}$ of
possible outputs, with $u\in\mathcal E_n(u)$. We say that two inputs
$u,u'\in\{0,1\}^{n}$ are \emph{confusable}, and write $u\sim_n u'$, if they
can produce a common output:
\[
  u\sim_n u'
  \quad\Longleftrightarrow\quad
  \mathcal E_n(u)\cap\mathcal E_n(u')\ne\varnothing.
\]
A code $\mathcal C\subseteq\{0,1\}^{n}$ corrects $\mathcal E_n$ if no two
distinct codewords in $\mathcal C$ are confusable.

\begin{definition}[Splittability]
\label{def:splittability}
Let $n=m+N$, and let $\mathcal R_N$ be a reflexive and symmetric relation
on $\{0,1\}^{N}$. We say that $\mathcal E_n$ is
\emph{$(m,N,\mathcal R_N)$-splittable} if for all $x,x'\in\{0,1\}^{m}$ and
$s,s'\in\{0,1\}^{N}$,
\begin{equation}
  (x\Vert s)\sim_n(x'\Vert s')
  \quad\Longrightarrow\quad
  (s,s')\in\mathcal R_N.
  \label{eq:splitting-property}
\end{equation}
\end{definition}

We then choose the map $A:\{0,1\}^{h}\to\{0,1\}^{N}$ to be an injective function such that 
\begin{equation}
  T\ne T'
  \quad\Longrightarrow\quad
  \bigl(A(T),A(T')\bigr)\notin\mathcal R_N,
  \label{eq:tag-code-separation}
\end{equation}
and refer to such a map $A$ as a \emph{tag code}. Then, it is easy to see that if two codewords $x\Vert A(T)$ and $x'\Vert A(T')$ are confusable, then they must have the same tag $T=T'$, and hence share the exact same suffix $s=A(T)$.

\subsection{Correctness and redundancy}
\label{subsec:code-correctness}

Since the tag code $A$ forces any two confusable codewords to share the same suffix $s=A(T)$, we only need to distinguish a payload $x\in\{0,1\}^{m}$ from other payloads $x'$ that can be confused with $x$ when both are followed by the same suffix $s$. As the suffix $s$ depends on the tag of $x$, we define the set of \emph{payload competitors} of $x$ by considering all possible common suffixes $s\in\{0,1\}^{N}$:
\begin{equation}
\Gamma_{\mathcal E}(x)=\bigl\{x'\in\{0,1\}^{m}\setminus\{x\}:(x\Vert s)\sim_n(x'\Vert s)\text{ for some }s\in\{0,1\}^{N}\bigr\}.\label{eq:payload-competitors}
\end{equation}
In our constructions, it is convenient to work with a slightly larger set
$\widetilde\Gamma(x)$ that is easier to list and satisfies
\begin{equation}
  \Gamma_{\mathcal E}(x)
  \subseteq\widetilde\Gamma(x)
  \subseteq\{0,1\}^{m}\setminus\{x\},
  \qquad |\widetilde\Gamma(x)|\le L.
  \label{eq:competitor-cover}
\end{equation}
Applying the construction of Section~\ref{sec:recycling} with
$\Gamma(x)=\widetilde\Gamma(x)$ gives the following lemma.

\begin{lemma}[From tags to codes]
\label{lem:tag-to-code}
Suppose that $\mathcal E_n$ is $(m,N,\mathcal R_N)$-splittable and
\eqref{eq:competitor-cover} holds. Let $b,r$ satisfy
\eqref{eq:seed-parameters}, let $h=2b+r$, and let
$A:\{0,1\}^{h}\to\{0,1\}^{N}$ be a tag code satisfying
\eqref{eq:tag-code-separation}. Then the encoder $\Enc$ in
\eqref{eq:protected-tag-encoder} is injective, its image $\mathcal C$
corrects $\mathcal E_n$, and the redundancy of $\mathcal C$ is
\begin{equation}
  n-\log|\mathcal C|=(r+2)+(N-h).
  \label{eq:code-redundancy}
\end{equation}
In particular, setting $r=\lceil\log(2L)\rceil$ gives redundancy at most
$\log L+4+(N-h)$.
\end{lemma}

\begin{proof}
Suppose that two codewords $x\Vert A(T)$ and $x'\Vert A(T')$ are confusable. By \eqref{eq:splitting-property}, we have $\bigl(A(T),A(T')\bigr)\in\mathcal R_N$, which by \eqref{eq:tag-code-separation} implies $T=T'$. Thus the two codewords share the same suffix $s=A(T)$. If $x\ne x'$, then $(x\Vert s)\sim_n(x'\Vert s)$ means $x'\in\Gamma_{\mathcal E}(x)\subseteq\widetilde\Gamma(x)$, which contradicts Theorem~\ref{thm:recycling}(a) since $T=T'$. Hence we must have $x=x'$, and Theorem~\ref{thm:recycling}(b) then implies that the additional message bits $(\sigma,\tau)$ are also equal. Thus distinct messages map to non-confusable codewords.

Since each message consists of $m+2b-2$ bits and each codeword has length $n=m+N$, the redundancy of $\mathcal C$ is
$n-\log|\mathcal C|=(m+N)-(m+2b-2)=(r+2)+(N-h)$, and the last claim follows from Theorem~\ref{thm:recycling}(c).
\end{proof}

\paragraph{Decoding.}
Given a received word $y$, let
$\mathcal E_n^{-1}(y)=\{z\in\{0,1\}^{n}:y\in\mathcal E_n(z)\}$ be the set of all length-$n$ inputs that can produce $y$. Suppose we can list $\mathcal E_n^{-1}(y)$, test membership in $A(\{0,1\}^h)$ and invert $A$ on its image. Since every candidate
$z\in\mathcal E_n^{-1}(y)$ has length $n=m+N$, we can split $z=x'\Vert u$ at
the known prefix length $|x'|=m$ without knowing the boundary in $y$. For each candidate $z=x'\Vert u\in\mathcal E_n^{-1}(y)$ such that $u=A(T')$ for some tag $T'=\varphi^{-1}(a')\Vert\varphi^{-1}(\lambda')\Vert\eta'$, we check whether
\begin{equation}
  \pi\bigl(\lambda'P_{x'}(a')\bigr)=\eta',
  \label{eq:decoder-hash-test}
\end{equation}
and stop at the first candidate $z$ that satisfies this equality. If a candidate passes the hash test, the decoder reconstructs the two good-seed sets in turn, checks that each seed lies among the first $K$ elements of its set, and reads off $(\sigma,\tau)$ from their ranks as in Section~\ref{sec:recycling}; if no candidate passes the hash test or if either rank check fails, the decoder outputs failure. %We then compute $\Good_1(p')$ and $\Good_2(p',a')$ using $\widetilde\Gamma(p')$, and recover $(\sigma,\tau)$ from the ranks of $a'$ and $\lambda'$ as in Section~\ref{sec:recycling}.

To see that this correctly recovers the transmitted codeword $c=x\Vert A(T)$, observe first that $c\in\mathcal E_n^{-1}(y)$ and $c$ passes the test \eqref{eq:decoder-hash-test}. If any other candidate $z=x'\Vert A(T')\in\mathcal E_n^{-1}(y)$ also passes, then $z\sim_n c$, so splittability and \eqref{eq:tag-code-separation} imply $T'=T$. If
$x'\ne x$, then $x'\in\Gamma_{\mathcal E}(x)\subseteq\widetilde\Gamma(x)$, and since the seeds $a,\lambda$ in $T$ are good for $x$, the equality \eqref{eq:decoder-hash-test} cannot hold for $x'$. Thus $z=c$ is the unique candidate in $\mathcal E_n^{-1}(y)$ that passes \eqref{eq:decoder-hash-test}, and its seed ranks give $(\sigma,\tau)$.

\subsection{An efficient construction}
\label{subsec:efficient-tag-protection}

It remains to construct a tag code $A:\{0,1\}^{h}\to\{0,1\}^{N}$ with small
overhead $N-h$. Let $D_N=\max_{s\in\{0,1\}^{N}}\bigl|\{s'\in\{0,1\}^{N}:(s,s')\in\mathcal R_N\}\bigr|$ be the maximum degree of the relation $\mathcal R_N$. A simple greedy procedure finds at least $2^N/D_N$ strings in $\{0,1\}^{N}$ that are pairwise unrelated under $\mathcal R_N$: we scan $\{0,1\}^{N}$ in lexicographic order, select the first unmarked string $s$, mark all strings $s'$ with $(s,s')\in\mathcal R_N$, and repeat. Since each step marks at most $D_N$ strings, we obtain a tag code $A$ of size $2^h$ whenever
\begin{equation}
  2^{N-h}\ge D_N.
  \label{eq:greedy-tag-condition}
\end{equation}
We can then store $A$ and $A^{-1}$ as lookup tables of size $O(2^N)$.

In our applications, the tag length is $h=O(\log n)$ and the degree $D_N$
is bounded by $CN^c$ for constants $C,c$. By \eqref{eq:greedy-tag-condition},
the extra redundancy needed to protect the tag is only $N-h=c\log\log n+O(1)$ bits. Furthermore, since $N=O(\log n)$, the greedy construction of $A$ runs in time $2^{O(\log n)}=\mathrm{poly}(n)$. We summarize this general construction in the following theorem.

\begin{theorem}[Efficient codes for splittable channels]
\label{thm:efficient-channel-code}
Let $\mathcal E$ be a channel, let $c\ge0$ and $C\ge1$ be constants, and let $L=L(n)\ge 1$ be an integer computable from $n$ in $\operatorname{poly}(n)$ time, with $\log L=O(\log n)$. Suppose that, for all sufficiently large $n$ and every split $n=m+N$ with $1\le N\le n/2$, the following hold. All listing algorithms below are uniform and deterministic.
\begin{enumerate}[label=(\alph*),leftmargin=*]
\item there is a reflexive symmetric relation $\mathcal R_N$ with $D_N\le CN^c$ such that $\mathcal E_n$ is
$(m,N,\mathcal R_N)$-splittable, and the $\mathcal R_N$-neighborhood of any
$s\in\{0,1\}^{N}$ can be listed in $\mathrm{poly}(N)$ time;
\item for each $x\in\{0,1\}^{m}$, a set $\widetilde\Gamma(x)$ satisfying
\eqref{eq:competitor-cover} can be listed in $\mathrm{poly}(n)$ time;
\item for every string $y$, the set $\mathcal E_n^{-1}(y)$ can be listed in $\operatorname{poly}(n+|y|)$ time.
\end{enumerate}
Then, for all sufficiently large $n$, there is an explicit code
$\mathcal C\subseteq\{0,1\}^{n}$ correcting $\mathcal E_n$ with redundancy
\begin{equation}
  n-\log|\mathcal C|
  \le \log L+c\log\log n+O(1),
  \label{eq:general-code-bound}
\end{equation}
with deterministic $\mathrm{poly}(n)$-time encoding and $\mathrm{poly}(n+|y|)$-time decoding.
\end{theorem}

\begin{proof}
Given $n$, we set $b=\lceil\log(2Ln)\rceil$, $r=\lceil\log(2L)\rceil$, and
$h=2b+r$. We then choose $N=h+\bigl\lceil c\log(2h)+\log C\bigr\rceil$ and $m=n-N$. Since $h=O(\log n)$, we have $N\le 2h\le n/2$ for large enough $n$. It is
easy to check that $b$ and $r$ satisfy \eqref{eq:seed-parameters}, and that $D_N\le CN^c\le C(2h)^c\le 2^{N-h}$, so \eqref{eq:greedy-tag-condition} holds and the greedy procedure yields the tag code $A$.

Since $2^b=O(Ln)$ and $2^N=2^h\cdot O(h^c)=O(L^3 n^2 \log^c n)$ are both polynomial in $n$, we can find an irreducible polynomial of degree $b$ over $\F_2$ and construct the lookup tables for $A$ and $A^{-1}$ in
$\mathrm{poly}(n)$ time using (a). By (b) and (c), the encoder \eqref{eq:protected-tag-encoder} runs in $\mathrm{poly}(n)$ and the decoder runs in $\mathrm{poly}(n+|y|)$ time. Finally, by Lemma~\ref{lem:tag-to-code}, the redundancy is $(r+2)+(N-h)\leq \log L+c\log\log n+O(1)$.
\end{proof}

In the next section, we verify that the Damerau--Levenshtein channel satisfies the conditions of Theorem~\ref{thm:efficient-channel-code} with $\mathcal R_N$ equal to the confusion relation $\sim_N$ of $\mathcal E_N$ itself. When a channel is self-splittable, any code correcting $\mathcal E_N$ can serve as the tag code $A$. In Appendix~D, we use our own code construction at length $N=O(\log n)$ as the tag code $A$, which reduces the precomputation time of $A$ from $\mathrm{poly}(n)$ to $\mathrm{polylog}(n)$.

% =====================================================================

\section{Codes in the Damerau--Levenshtein metric}
\label{sec:dl-codes}

A \textit{Damerau--Levenshtein (DL)} operation is an insertion, deletion, substitution, or adjacent transposition $u_i u_{i+1}\mapsto u_{i+1} u_i$, and $d_{\mathrm{DL}}(u,v)$ is the minimum number of DL operations transforming $u$ into $v$~\cite{Dam64}. Every operation is undone by one of the allowed operations, so $d_{\mathrm{DL}}$ is a graph metric, and $d_{\mathrm{DL}}\le d_{\mathrm{L}}\le 2d_{\mathrm{DL}}$ since a transposition is two substitutions. Let $B_q(u)=\{v\in\{0,1\}^{*}:d_{\mathrm{DL}}(u,v)\le q\}$. As for edits, a code corrects $t$ DL errors (the channel
$\mathcal E_n(u)=B_t(u)$) if and only if its codewords are at pairwise DL
distance at least $2t+1$, so
\begin{equation}
  u\sim_n u'
  \quad\Longleftrightarrow\quad
  d_{\mathrm{DL}}(u,u')\le 2t.
  \label{eq:dl-confusability}
\end{equation}

\subsection{Splitting at a fixed boundary}
\label{subsec:dl-splitting}

To show that the DL channel is splittable, we track a boundary marker
between the payload and the suffix through a sequence of DL operations.

\begin{lemma}[Projection at a fixed boundary]
\label{lem:dl-projection}
For all $x,x'\in\{0,1\}^{m}$ and $s,s'\in\{0,1\}^{N}$,
\begin{equation}
  \max\bigl\{d_{\mathrm{DL}}(x,x'),d_{\mathrm{DL}}(s,s')\bigr\}
  \le d_{\mathrm{DL}}(x\Vert s,x'\Vert s').
  \label{eq:dl-projections}
\end{equation}
\end{lemma}

\begin{proof}
Let $k=d_{\mathrm{DL}}(x\Vert s,x'\Vert s')$, and let
$w_0=x\Vert s,w_1,\ldots,w_k=x'\Vert s'$ be any sequence of $k$ operations. We maintain a factorization $w_i=\ell_i\Vert\rho_i$, starting from $(\ell_0,\rho_0)=(x,s)$. An operation confined to one side of the marker is applied to that part (with an insertion at the marker appended to $\ell_i$). A transposition straddling the marker swaps the last symbol of $\ell_i$ with the first symbol of $\rho_i$, which changes $\rho_i$ by at most one substitution and leaves $|\ell_i|$ unchanged.

By induction on $i$, we have the invariant $d_{\mathrm{DL}}(s,\rho_i)+\bigl||\ell_i|-m\bigr|\le i$. Indeed, each operation on the left changes $|\ell_i|$ by at most one and leaves $\rho_i$ unchanged, while each operation on the right or straddling the marker leaves $|\ell_i|$ unchanged and changes $\rho_i$ by at most one DL operation.

At $i=k$, write $x'\Vert s'=\ell_k\Vert\rho_k$ and $j=|\ell_k|$. Whether $j\ge m$ or $j<m$, the suffix $s'$ is obtained from $\rho_k$ by at most $|j-m|$ insertions or deletions at its beginning, so $d_{\mathrm{DL}}(s,s')\le d_{\mathrm{DL}}(s,\rho_k)+|j-m|\le k$.

Finally, since string reversal preserves $d_{\mathrm{DL}}$ and swaps prefixes with suffixes, applying the same bound to the reversed strings
gives $d_{\mathrm{DL}}(x,x')\le k$.
\end{proof}

\begin{example}
Let $m=N=3$, $x=011$, and $s=001$. Applying a straddling transposition at
positions $3,4$, deleting the first bit, and appending $0$ gives
\[
  011\Vert 001
  \ \longrightarrow\ 010\Vert 101
  \ \longrightarrow\ 10\Vert 101
  \ \longrightarrow\ 10\Vert 1010,
\]
where the bar tracks $\ell_i\Vert\rho_i$. Splitting the final word at
$m=3$ gives $x'=101$ and $s'=010$. Here $\rho_3=1010$ differs from $s$ by
two operations and $|\ell_3|=2$, so $s'$ is obtained from $\rho_3$ by
deleting its first bit and $d_{\mathrm{DL}}(s,s')\le 2+|2-3|=3$.
\end{example}

Taking $\mathcal R_N$ to be the confusion relation
$d_{\mathrm{DL}}(s,s')\le 2t$ on $\{0,1\}^{N}$, Lemma~\ref{lem:dl-projection}
shows that the $t$-DL channel is self-splittable.

\subsection{Enumerating the competitors}
\label{subsec:dl-competitors}

Lemma~\ref{lem:dl-projection} also pins down the exact payload competitor
set $\Gamma_{\mathcal E}(x)$ from \eqref{eq:payload-competitors}.

\begin{corollary}[Payload competitors]
\label{cor:dl-competitors}
For any $x,x'\in\{0,1\}^{m}$ and $s\in\{0,1\}^{N}$, we have
$d_{\mathrm{DL}}(x\Vert s,x'\Vert s)=d_{\mathrm{DL}}(x,x')$. Consequently,
for the $t$-DL channel,
\begin{equation}
  \Gamma_{\mathcal E}(x)
  =\Gamma_{\mathrm{DL}}(x)
  :=\bigl(B_{2t}(x)\cap\{0,1\}^{m}\bigr)\setminus\{x\}.
  \label{eq:dl-exact-competitors}
\end{equation}
\end{corollary}

\begin{proof}
Lemma~\ref{lem:dl-projection} gives
$d_{\mathrm{DL}}(x,x')\le d_{\mathrm{DL}}(x\Vert s,x'\Vert s)$, and the
reverse inequality follows by applying a shortest operation sequence from
$x$ to $x'$ to the prefix of $x\Vert s$. Equation~\eqref{eq:dl-exact-competitors}
then follows from \eqref{eq:payload-competitors} and
\eqref{eq:dl-confusability}.
\end{proof}

A string of length $\ell'$ admits $2(\ell'+1)$ insertions, $\ell'$ deletions, $\ell'$ substitutions, and at most $\ell'$ adjacent transpositions, for a total of at most $5\ell'+2$ operations. Thus, for any $u\in\{0,1\}^{\ell}$, applying all sequences of at most $q$ operations lists $B_q(u)$ in at most $M(\ell,q)=(q+1)\bigl(5(\ell+q)+2\bigr)^q$ entries and time $O\bigl(M(\ell,q)(\ell+q+1)\bigr)$. Assuming $n\ge 2t+1$,
this gives all three lists required in Section~\ref{sec:tags-to-codes}:
\begin{itemize}[nosep,leftmargin=*]
\item \emph{Competitors:} Filtering the list of $B_{2t}(x)$ for strings of
length $m$ different from $x$ lists $\Gamma_{\mathrm{DL}}(x)$ in time $O_t(n^{2t+1})$ with at most $L=D_t^{\mathrm{DL}}n^{2t}$ and $D_t^{\mathrm{DL}}=(2t+1)10^{2t}$ entries.
\item \emph{Suffix neighborhoods:} For $N\ge 2t+1$, the
$\mathcal R_N$-neighborhood $B_{2t}(s)\cap\{0,1\}^{N}$ has size at most
$D_N\le(2t+1)(10N)^{2t}$. For $1\le N\le 2t$, we have $D_N\le 2^N\le 2^{2t}\leq D_t^{\mathrm{DL}}N^{2t}$.
\item \emph{Preimages:} For any received word $y$ of length in
$[n-t,n+t]$, the preimage set $\mathcal E_n^{-1}(y)=B_t(y)\cap\{0,1\}^{n}$
is listed with $O_t(n^t)$ entries in time $O_t(n^{t+1})$.
\end{itemize}

\subsection{The resulting codes}
\label{subsec:dl-result}

Setting $b=\lceil\log(2Ln)\rceil$, $r=\lceil\log(2L)\rceil$, and $h=2b+r$
with $L=D_t^{\mathrm{DL}}n^{2t}$, and applying the encoder
\eqref{eq:protected-tag-encoder} with either a greedy or a bootstrapped tag
code $A:\{0,1\}^{h}\to\{0,1\}^{N}$, we obtain our main result (Theorem~\ref{thm:main}).

\begin{theorem}[Damerau--Levenshtein codes]
\label{thm:dl-codes}
For every fixed $t\ge1$ and all sufficiently large $n$, there are explicit
binary codes $\mathcal C\subseteq\{0,1\}^{n}$ correcting any combination of
at most $t$ insertions, deletions, substitutions, and adjacent
transpositions, in two variants:
\begin{enumerate}[label=(\alph*),leftmargin=*]
\item (bootstrapped) with precomputation time $\operatorname{polylog}(n)$
and redundancy $n-\log|\mathcal C|\le 2t\log n+2t\log\log n+O_t(\log\log\log n)$;
\item (greedy) with precomputation time $n^{6t+2+o(1)}$ and redundancy
$n-\log|\mathcal C|\le 2t\log n+2t\log\log n+O_t(1)$.
\end{enumerate}
In both variants, encoding and full-message decoding take
$O_t(n^{2t+2}\log^2 n)$ bit operations, and the payload (the first
$n-O_t(\log n)$ bits of the codeword) is recovered in $O_t(n^{t+1}\log n)$
bit operations.
\end{theorem}

\begin{proof}
By Lemma~\ref{lem:dl-projection}, Corollary~\ref{cor:dl-competitors}, and
Lemma~\ref{lem:tag-to-code}, the code $\mathcal C$ corrects $t$ DL errors
and has redundancy
\begin{equation}
  n-\log|\mathcal C|=(r+2)+(N-h),
  \qquad r+2\le 2t\log n+\log D_t^{\mathrm{DL}}+4.
  \label{eq:dl-redundancy-split}
\end{equation}

For variant (b), we set $N=h+\lceil 2t\log(20h)+\log(2t+1)\rceil$. For
large $n$, we have $2t+1\le N\le 2h$, so
$D_N\le(2t+1)(10N)^{2t}\le(2t+1)(20h)^{2t}\le 2^{N-h}$. Since
$h=(6t+2)\log n+O_t(1)$, the greedy construction of
Section~\ref{subsec:efficient-tag-protection} yields a tag code $A$ with $N-h=2t\log\log n+O_t(1)$ in time $O_t(2^N N^{2t+1})=n^{6t+2+o(1)}$. The outer field can be constructed in $\operatorname{poly}(b)=\operatorname{polylog}(n)$ time by Appendix~C.1, so the total precomputation time for variant (b) remains $n^{6t+2+o(1)}$.

For variant (a), since the DL channel is self-splittable, we can apply
bootstrapping (Lemma~D.1) to the greedy family at length
$N_h\le 2h=O_t(\log n)$. This yields a tag code $A_h$ with overhead
$N_h-h\le 2t\log h+2t\log\log h+O_t(1)=2t\log\log n+O_t(\log\log\log n)$
and precomputation time $(2h)^{6t+2+o(1)}=\operatorname{polylog}(n)$.

Finally, let $\tau_A$ be the time per tag-code operation ($\tau_A=O_t(\log n)$
for (b) and $\tau_A=\operatorname{polylog}(n)$ for (a)). By Lemma~C.5,
computing the tag or recovering the seed ranks from $L=O_t(n^{2t})$
competitors takes $O_t(n^{2t+2}\log^2 n)$ bit operations. During decoding,
testing each of the $O_t(n^t)$ candidates in $\mathcal E_n^{-1}(y)$ takes
$\tau_A$ time for the suffix and $O_t(n\log n)$ bit operations for Horner's
rule, giving $O_t(n^{t+1}\log n)$ bit operations to recover the payload and
$O_t(n^{2t+2}\log^2 n)$ bit operations to recover the full message.
\end{proof}

% =====================================================================
\section{Discussion and open problems}
\label{sec:discussion}

Our codes match the leading term $2t\log n$ of the greedy Gilbert--Varshamov bound for every fixed $t\ge 1$, resolving the long-standing factor-of-two bottleneck for mixed synchronization errors. Because our primary focus is explicitness, the construction leaves several compelling open problems.

\paragraph{Explicit codes matching the $(2t-1)\log n$ existential bound.}
For every fixed $t\ge 2$, En~Gad~\cite{EnGad26} recently proved non-constructively that $(2t-1)\log n+O_t(\log\log n)$ redundancy suffices for deletions, for general edits, and for Damerau--Levenshtein errors. In our framework, matching this existential bound would require reducing the net tag cost from $\log L+O(1)$ to roughly $\frac{2t-1}{2t}\log L+O(1)$ on the confusion graph. Can seed recycling be combined with the substring-counting invariants of~\cite{EnGad26} to construct explicit codes with redundancy $(2t-1)\log n+o(\log n)$? Even for two deletions, constructing explicit codes with redundancy $3\log n+o(\log n)$ remains open.

\paragraph{Optimal redundancy beyond $2t-1$.}
While En~Gad's results~\cite{EnGad26} lower the existential upper bound from $2t\log n$ to $(2t-1)\log n$ for all three error models, a significant gap remains to the sphere-packing lower bound of $t\log n-O_t(1)$. Can the existential leading coefficient be reduced below $2t-1$ for any of these channels? Determining the true optimal redundancy for constant $t\ge 2$ errors remains a fundamental challenge in combinatorial coding theory.

\paragraph{Efficiency and broader channel models.}
Like previous explicit constructions based on syndrome compression or graph coloring~\cite{SGB20syndrome,LGF25}, our encoding and decoding times contain $t$ in the exponent, stemming from the explicit enumeration of the ball $B_{2t}(x)$ to form the competitor list. A major algorithmic challenge is to construct explicit codes matching the $2t\log n$ bound with encoding and decoding times $f(t)\cdot n^{O(1)}$ whose polynomial degree is independent of $t$. Finally, it is natural to investigate whether seed recycling can be extended to the related channels discussed in Section~\ref{sec:intro}, including $q$-ary deletion channels, permutation codes, and burst errors.

\clearpage
\appendix

\section{Summary tables and notation}
\label{app:notation}

Tables~\ref{tab:deletions}--\ref{tab:dl} provide a fuller comparison of our codes (variant~(a) of Theorem~\ref{thm:main}, which has deterministic $\operatorname{poly}(n)$-time encoding and decoding after $\operatorname{polylog}(n)$-time precomputation) with prior explicit constructions, existential upper bounds, and sphere-packing lower bounds for $t$ deletions, $t$ edits, and $t$ Damerau--Levenshtein errors. Bounds marked $\dagger$ are consequences of the coloring framework of~\cite{LGF25}. Table~\ref{tab:notation} summarizes the main global notation of the paper. Throughout, logarithms are to base two, $t\ge 1$ is fixed, and constants in $O_t(\cdot)$ depend only on $t$.

\begin{table}[ht]
\centering
\footnotesize
\renewcommand{\arraystretch}{1.25}
\setlength{\aboverulesep}{0pt}
\setlength{\belowrulesep}{0pt}
\caption{Binary codes correcting $t$ deletions, or equivalently $t$ insertions, or any combination of $t$ insertions and deletions~\cite{Lev66}, for fixed $t$ and $n\to\infty$.}
\label{tab:deletions}
\begin{tabular}{lll}
\toprule
\textbf{Errors} & \textbf{Redundancy} & \textbf{Reference} \\
\midrule
\multicolumn{3}{l}{\emph{Explicit constructions}} \\
$1$ & $\log(n+1)$ & \cite{VT65,Lev66} \\
$2$ & $8\log n+O(\log\log n)$;\ \ $7\log n+o(\log n)$ & \cite{GS19};\ \ \cite{SRB20} \\
$2$ & $4\log n+10\log\log n+O(1)$ & \cite{GH21} \\
$t$ & $O(t^2\log t\log n)$;\ \ $8t\log n+o(\log n)$ & \cite{BGZ18};\ \ \cite{SB21} \\
$t$ & $4t\log n+o(\log n)$;\ \ $(4t-1)\log n+o(\log n)$ & \cite{SGB20opt};\ \ \cite{SPCH21,SPCH22} \\
\rowcolor{LightCyan} $t$ & $\boldsymbol{2t}\log n+2t\log\log n+O_t(\log\log\log n)$ & \textbf{this paper} \\
\midrule
\multicolumn{3}{l}{\emph{Existential bounds}} \\
$t$ & $2t\log n+O_t(1)$ (GV-type);\ \ $2t\log n-\log\log n+O_t(1)$ & \cite{Lev66,Lev02};\ \ \cite{ABGHK24} \\
$t\ge 2$ & $(2t-1)\log n+O_t(\log\log n)$ & \cite{EnGad26} \\
\midrule
\emph{Lower bound} & $t\log n-O_t(1)$ & \cite{Lev66,Lev02} \\
\bottomrule
\end{tabular}
\end{table}

\begin{table}[ht]
\centering
\footnotesize
\renewcommand{\arraystretch}{1.25}
\setlength{\aboverulesep}{0pt}
\setlength{\belowrulesep}{0pt}
\caption{Binary codes correcting edits (insertions, deletions, and substitutions). A row ``$t$ deletions $+$ $s$ substitutions'' guarantees correction of up to $t$ deletions and up to $s$ substitutions, which is weaker than correcting any $t+s$ edits.}
\label{tab:edits}
\begin{tabular}{lll}
\toprule
\textbf{Errors} & \textbf{Redundancy} & \textbf{Reference} \\
\midrule
\multicolumn{3}{l}{\emph{Explicit constructions}} \\
$1$ edit & $\log n+1$ & \cite{Lev66} \\
$1$ deletion $+$ $1$ substitution & $6\log n+O(1)$ & \cite{SWWY23} \\
$t$ deletions $+$ $s$ substitutions & $(4t+3s)\log n+o(\log n)$ & \cite{SPCH22} \\
$2$ edits & $6\log n+O(\log\log n)$;\ \ $6\log n+8$ & \cite{SG24};\ \ \cite{PZ24} \\
$t$ edits & $4t\log n+o(\log n)$;\ \ $4t\log n+O(\log\log n)$ & \cite{SGB20opt,SGB20syndrome};\ \ \cite{LGF25} \\
$t$ edits & $(4t-1)\log n+O_t(\log\log n)$ & \cite{LGF25}$^{\dagger}$ \\
\rowcolor{LightCyan} $t$ edits & $\boldsymbol{2t}\log n+2t\log\log n+O_t(\log\log\log n)$ & \textbf{this paper} \\
\midrule
\multicolumn{3}{l}{\emph{Existential bounds}} \\
$t$ edits & $2t\log n+O_t(1)$ (GV-type) & \cite{Lev66,Lev02} \\
$t\ge 2$ edits & $(2t-1)\log n+O_t(\log\log n)$ & \cite{EnGad26} \\
\midrule
\emph{Lower bound} & $t\log n-O_t(1)$ & \cite{Lev66,Lev02} \\
\bottomrule
\end{tabular}
\end{table}

\begin{table}[ht]
\centering
\footnotesize
\renewcommand{\arraystretch}{1.25}
\setlength{\aboverulesep}{0pt}
\setlength{\belowrulesep}{0pt}
\caption{Binary codes in the Damerau--Levenshtein metric (insertions, deletions, substitutions, and adjacent transpositions).}
\label{tab:dl}
\begin{tabular}{lll}
\toprule
\textbf{Errors} & \textbf{Redundancy} & \textbf{Reference} \\
\midrule
\multicolumn{3}{l}{\emph{Explicit constructions}} \\
$1$ deletion $+$ $\ell$ adjacent transpositions & $2\ell\log n+\log(n+2\ell+1)$ & \cite{GYM18} \\
$t$ DL errors & $8t\log n+O_t(\log\log n)$ & \cite{LGF25} (via $2t$-edits code) \\
$t$ DL errors & $4t\log n+O_t(\log\log n)$ & \cite{LGF25}$^{\dagger}$ \\
\rowcolor{LightCyan} $t$ DL errors & $\boldsymbol{2t}\log n+2t\log\log n+O_t(\log\log\log n)$ & \textbf{this paper} \\
\midrule
\multicolumn{3}{l}{\emph{Existential bounds}} \\
$1$ deletion or $1$ adjacent transposition
& $\log n + O(\log\log n)$
& \cite{GGRW23} \\
$t$ DL errors & $2t\log n+O_t(1)$ (GV-type) & Section~\ref{sec:intro} \\
$t\ge 2$ DL errors & $(2t-1)\log n+O_t(\log\log n)$ & \cite{EnGad26} \\
\midrule
\emph{Lower bound} & $t\log n-O_t(1)$ & \cite{Lev66,Lev02} \\
\bottomrule
\end{tabular}
\end{table}

\begin{table}[ht]
\centering
\footnotesize
\renewcommand{\arraystretch}{1.2}
\caption{Summary of main notation.}
\label{tab:notation}
\begin{tabular*}{\textwidth}{@{\extracolsep{\fill}}p{0.28\textwidth}p{0.68\textwidth}@{}}
\toprule
\textbf{Symbol} & \textbf{Meaning} \\
\midrule
$t,\ 2t$ & Number of errors per codeword (fixed), and confusability radius \\
$n=m+N$ & Codeword length $=$ payload length $m$ $+$ protected suffix length $N$ \\
$d_{\mathrm L}(u,v),\ d_{\mathrm{DL}}(u,v)$ & Edit (Levenshtein) distance and Damerau--Levenshtein distance \\
$\mathcal E_q(u),\ B_q(u)$ & Edit ball and DL ball of radius $q$ around $u$ \\
$\mathcal E=(\mathcal E_n),\ \mathcal E_n^{-1}(y),\ \sim_n$ & Channel output map, preimage set of output $y$, and confusability relation \\
$\mathcal R_N,\ D_N$ & Suffix relation in Definition~\ref{def:splittability} and its maximum degree \\
$\Gamma(x),\ \Gamma_{\mathcal E}(x),\ \widetilde\Gamma(x)$ & Abstract competitor set, exact payload competitors, and listable superset \\
$\Gamma_{\mathrm{DL}}(x)$ & Exact DL payload competitors $(B_{2t}(x)\cap\{0,1\}^{m})\setminus\{x\}$ \\
$L,\ D_t^{\mathrm{DL}}$ & Competitor bound $|\Gamma(x)|\le L$; for DL errors, $L=D_t^{\mathrm{DL}}n^{2t}$ with $D_t^{\mathrm{DL}}=(2t+1)10^{2t}$ \\
$b,\ Q=2^b,\ K=2^{b-1}$ & Field bit-length, field size $|\mathbb F_Q|$, and number of selectable good seeds per stage \\
$k=\lceil m/b\rceil,\ P_x(Z)$ & Number of $b$-bit blocks in $x$ and the associated polynomial over $\mathbb F_Q$ \\
$\varphi,\ \iota,\ \pi,\ r$ & Bijection $\{0,1\}^{b}\to\mathbb F_Q$, canonical order $\iota(a)=\val(\varphi^{-1}(a))$, and projection $\pi:\mathbb F_Q\to\{0,1\}^{r}$ \\
$\Good_1(x),\ \Good_2(x,a)$ & Good evaluation seeds $a\in\mathbb F_Q$ and good multipliers $\lambda\in\mathbb F_Q$ (complements $\Bad_1,\Bad_2$) \\
$\Tag(x,\sigma,\tau),\ h=2b+r$ & Recycled tag carrying $2b-2$ extra message bits $(\sigma,\tau)\in(\{0,1\}^{b-1})^2$, and its length \\
$A:\{0,1\}^{h}\to\{0,1\}^{N}$ & Tag code protecting the $h$-bit tag in an $N$-bit suffix, with overhead $N-h$ \\
\bottomrule
\end{tabular*}
\end{table}

\subsection{Derivation of the bounds marked $\dagger$}
\label{app:lgf-consequences}

Here we briefly record how the two bounds marked $\dagger$ in Tables~ \ref{tab:edits}, and~\ref{tab:dl} follow from the coloring framework of Li, Gabrys, and Farnoud~\cite{LGF25}. Throughout, $t\ge 2$ is fixed, $m$ is the payload length, and $n$ is the codeword length. In each derivation below, $h$ denotes the length of the syndrome under consideration, independently of the seed-recycling parameters.

\paragraph{General edits.}
Let $\Phi_1(x)$ be Levenshtein's single-edit syndrome of length $\log m+O(1)$~\cite{Lev66}. Applying~\cite[Theorem~5]{LGF25} to extend this single-edit syndrome to correction of \(t\) edits gives an additional syndrome $\Phi_{t\mid 1}(x)$ of length $(4t-2)\log m+O_t(\log\log m)$ such that $S(x)=\Phi_1(x)\Vert\Phi_{t\mid 1}(x)$ separates $x\in\{0,1\}^m$ from all strings within edit distance $2t$, with total length $h=|S(x)|\le(4t-1)\log m+O_t(\log\log m)$. Applying~\cite[Theorem~3]{LGF25} at length $h$ yields a secondary syndrome $\Psi:\{0,1\}^h\to\{0,1\}^\rho$ of length $\rho=O_t(\log h)$, and~\cite[Lemma~37]{LGF25} shows that $x\mapsto x\Vert S(x)\Vert\operatorname{Rep}_{2t+1}(\Psi(S(x)))$ (where $\operatorname{Rep}_{2t+1}$ repeats each bit $2t+1$ times) is a polynomial-time encodable and decodable $t$-edit-correcting code of redundancy $n-m=h+(2t+1)\rho\le(4t-1)\log n+O_t(\log\log n)$.

\paragraph{Damerau--Levenshtein errors.}
Every DL operation is a $2$-substring edit (an adjacent transposition replaces a length-$2$ substring by its reversal). Thus, under the noiseless-syndrome convention of~\cite[Remark~1]{LGF25}, 
specializing ~\cite[Theorem~6]{LGF25} to $t$ $2$-substring edits gives a polynomial-time computable syndrome $\Phi_{\mathrm{DL}}:\{0,1\}^m\to\{0,1\}^h$ of length $h=4t\log m+O_t(\log\log m)$ that assigns distinct syndromes to payloads at DL distance at most $2t$. Encoding $\Phi_{\mathrm{DL}}(x)$ with the greedy DL tag code $A:\{0,1\}^h\to\{0,1\}^N$ of Section~\ref{subsec:efficient-tag-protection} (which has overhead $N-h=O_t(\log h)=O_t(\log\log m)$ and minimum DL distance at least $2t+1$) gives the codeword $c(x)=x\Vert A(\Phi_{\mathrm{DL}}(x))$. By Lemma~\ref{lem:dl-projection}, $d_{\mathrm{DL}}(c(x),c(x'))\le 2t$ implies both $d_{\mathrm{DL}}(x,x')\le 2t$ and $d_{\mathrm{DL}}(A(\Phi_{\mathrm{DL}}(x)),A(\Phi_{\mathrm{DL}}(x')))\le 2t$, which forces $\Phi_{\mathrm{DL}}(x)=\Phi_{\mathrm{DL}}(x')$ and hence $x=x'$. The resulting code has redundancy $N=4t\log n+O_t(\log\log n)$ and polynomial-time encoding and decoding.

\section{Additional proofs and decoding details for Sections~\ref{sec:recycling} and~\ref{sec:tags-to-codes}}
\label{app:section23}

We first record the direct consequence of Section~\ref{sec:recycling} when the tag is transmitted over a noiseless side channel, and then give the formal specification of the generic decoder from Section~\ref{sec:tags-to-codes}, including its behavior on arbitrary (possibly invalid) inputs needed for bootstrapping in Appendix~\ref{sec:bootstrap}.

\begin{corollary}[Recovery from a noiseless tag]
\label{cor:sync}
Let $\mathcal E_m$ be a channel on $\{0,1\}^m$ whose confusion graph has maximum degree at most $L\ge 1$. Let $\Gamma(x)$ be the neighborhood of $x$ in this graph, and set $b=\lceil\log(2Lm)\rceil$ and $r=\lceil\log(2L)\rceil$. If we transmit $x\in\{0,1\}^m$ through $\mathcal E_m$ and $\Tag(x,\sigma,\tau)$ (where $\sigma,\tau\in\{0,1\}^{b-1}$) over a noiseless side channel, then $(x,\sigma,\tau)$ can be uniquely recovered from any output $y\in\mathcal E_m(x)$ and the tag, with net redundancy $(2b+r)-(2b-2)=r+2<\log L+4$. Moreover, suppose that $L=m^{O(1)}$ is computable in $\poly(m)$ time, neighborhoods can be listed in $\poly(m)$ time, and $\mathcal E_m^{-1}(y)$ can be listed in $\poly(m+|y|)$ time. Then, encoding takes $\poly(m)$ time and decoding takes $\poly(m+|y|)$ time.
\end{corollary}

\begin{proof}
The parameters $b,r$ satisfy \eqref{eq:seed-parameters}. Given $y\in\mathcal E_m(x)$ and the tag $T=\varphi^{-1}(a)\Vert\varphi^{-1}(\lambda)\Vert\eta$, the decoder scans $z\in\mathcal E_m^{-1}(y)$ and checks whether $\pi(\lambda P_z(a))=\eta$. The true input $x$ passes this check, whereas every other $z\in\mathcal E_m^{-1}(y)$ is confusable with $x$ (so $z\in\Gamma(x)$) and therefore fails because $\lambda\in\Good_2(x,a)$ ensures $\pi(\lambda P_z(a))\ne\pi(\lambda P_x(a))=\eta$. Once $x$ is identified, the decoder reconstructs $\Good_1(x)$ and $\Good_2(x,a)$ and recovers $(\sigma,\tau)$ from the ranks of $a$ and $\lambda$. The net cost and running times follow from Theorem~\ref{thm:recycling} and Section~\ref{sec:recycling}.
\end{proof}

\paragraph{The generic decoder.}
Fix the parameters and tag code $A:\{0,1\}^h\to\{0,1\}^N$ of Lemma~\ref{lem:tag-to-code}, let $\mathcal C$ be the image of $\Enc$ in \eqref{eq:protected-tag-encoder}, and let $K=2^{b-1}$. The good- and bad-seed sets below are defined with respect to the same competitor sets $\Gamma(x)=\widetilde\Gamma(x)$ used by the encoder (implemented via the procedures of Appendix~\ref{sec:computing-tags}).

Assume that, on every input, we can list $\mathcal E_n^{-1}(y)$ and $\widetilde\Gamma(x)$, test membership in $A(\{0,1\}^h)$, and invert $A$ on its image. All these algorithms are assumed to terminate, and the lists may contain repetitions.

\begin{procedure}[Generic decoder]
\label{proc:generic-dec}
Given a string $y$, perform the following steps:
\begin{enumerate}[label=\arabic*.,leftmargin=*]
  \item \emph{Identify a candidate.} List $\mathcal E_n^{-1}(y)$. For each candidate $z=x'\Vert u$ with $|x'|=m$, check whether $u\in A(\{0,1\}^h)$. If so, write $A^{-1}(u)=\alpha\Vert\gamma\Vert\eta$ with $|\alpha|=|\gamma|=b$ and $|\eta|=r$, set $a=\varphi(\alpha)$ and $\lambda=\varphi(\gamma)$, and stop at the first candidate $z$ satisfying $\pi(\lambda P_{x'}(a))=\eta$. If no candidate passes, return \emph{failure}. (To recover only the payload, output $x'$ immediately after Step~1.)

  \item \emph{Recover the first seed rank.} Compute $\Bad_1(x')$. If $a\notin\Good_1(x')$, or if the rank $v$ of $a$ in $\Good_1(x')$ exceeds $K$, return \emph{failure}.

  \item \emph{Recover the second seed rank.} Compute $\Bad_2(x',a)$. If $\lambda\notin\Good_2(x',a)$, or if the rank $w$ of $\lambda$ in $\Good_2(x',a)$ exceeds $K$, return \emph{failure}. Otherwise, output $\bigl(x',\bin_{b-1}(v-1),\bin_{b-1}(w-1)\bigr)$.
\end{enumerate}
\end{procedure}

\begin{lemma}[Correctness and termination of the generic decoder]
\label{lem:generic-dec}
Under the hypotheses of Lemma~\ref{lem:tag-to-code} and the algorithmic assumptions above, for any codeword $c=\Enc(x,\sigma,\tau)$ and any valid channel output $y\in\mathcal E_n(c)$, the codeword $c$ is the unique candidate passing Step~1 of Procedure~\ref{proc:generic-dec}, and the procedure outputs $(x,\sigma,\tau)$. On an arbitrary input $y\in\{0,1\}^*$, Procedure~\ref{proc:generic-dec} always terminates, computes each bad-seed set at most once, and outputs either a message in $\{0,1\}^{m+2b-2}$ or failure. With the parameter choices, tag code, and listing algorithms used in the proof of Theorem~\ref{thm:efficient-channel-code}, the procedure runs in $\poly(n+|y|)$ time.
\end{lemma}

\begin{proof}
Correctness on valid channel outputs $y\in\mathcal E_n(c)$ was shown in Section~\ref{subsec:code-correctness}: $c=x\Vert A(T)$ passes Step~1, any other passing candidate $z=x'\Vert A(T')$ satisfies $z\sim_n c$ and hence $T'=T$, and $x'\ne x$ would place $x'\in\Gamma_{\mathcal E}(x)\subseteq\widetilde\Gamma(x)$, contradicting $\lambda\in\Good_2(x,a)$. Since the seeds $a,\lambda$ of $c$ have ranks $1+\val(\sigma)\le K$ and $1+\val(\tau)\le K$, Steps~2 and~3 output $(x,\sigma,\tau)$. On an arbitrary input $y$, Step~1 scans a finite list, Step~2 computes $\Bad_1(x')$ only for the single candidate selected in Step~1, and Step~3 computes $\Bad_2(x',a)$ only after Step~2 has verified that $a\in\Good_1(x')$ (so that $P_{x'}(a)-P_{x''}(a)\ne 0$ for all $x''\in\widetilde\Gamma(x')$). Thus every step is well-defined and the procedure terminates. With the choices made in the proof of Theorem~\ref{thm:efficient-channel-code}, the listing, tag-code, and rank computations take $\poly(n+|y|)$ time. For the $t$-edit and $t$-DL channels, inputs of length outside $[n-t,n+t]$ are rejected immediately after reading at most $n+t+1$ bits, so the running time is $\poly(n)$ on all inputs.
\end{proof}

\paragraph{Testing codeword membership.}
To decide whether a given string $u\in\{0,1\}^n$ belongs to the code $\mathcal C$, run the full decoder and accept if and only if it returns a message $\mu$ satisfying $\Enc(\mu)=u$. Indeed, if $u\in\mathcal C$, then $u\in\mathcal E_n(u)$, so Lemma~\ref{lem:generic-dec} ensures that the decoder returns $\mu=\Enc^{-1}(u)$; conversely, if $\Enc(\mu)=u$, then $u\in\mathcal C$ by definition. This re-encoding test is used in Appendix~\ref{sec:bootstrap} to invert the inner tag code.

\section{Encoding and decoding}
\label{sec:algorithms}

In this appendix, we give deterministic implementations of the tag construction from Section~\ref{sec:recycling} and the generic decoder from Appendix~\ref{app:section23}. Rather than scanning the entire field $\mathbb F_Q$ to find good seeds, we compute the bad-seed sets directly by finding roots of difference polynomials in the first stage and enumerating scaled copies of $\ker\pi$ in the second stage. In Section~\ref{subsec:edit-decoding}, we also describe a decoding optimization for channels with only insertions, deletions, and substitutions, which recovers the short suffix before enumerating payload candidates.

\subsection{Representation and precomputation}
\label{sec:precomp}

All algorithms are deterministic, and costs are measured in bit operations on a RAM. We fix a monic irreducible polynomial $f_b\in\mathbb F_2[U]$ of degree $b$ and represent $\mathbb F_Q=\mathbb F_2[U]/(f_b)$ (where $Q=2^b$) via the bijection $\varphi(u_1\cdots u_b)=\sum_{i=1}^{b}u_iU^{i-1}$, ordered by $\iota(a)=\val(\varphi^{-1}(a))$. Addition and comparison in $\mathbb F_Q$ take $O(b)$ bit operations, while schoolbook multiplication and inversion take $O(b^2)$ bit operations. The polynomial $f_b$ can be found deterministically in $\poly(b)$ time~\cite{Sho90}, which is $\polylog(n)$ time for our parameters.

For the tag code $A:\{0,1\}^h\to\{0,1\}^N$, let $\tau_A$ denote the time required to compute $A(T)$, or to test membership in $A(\{0,1\}^h)$ and compute $A^{-1}(u)$. We have $\tau_A=O_t(\log n)$ for the greedy table-lookup construction of Section~\ref{subsec:efficient-tag-protection} and $\tau_A=\polylog(n)$ for the bootstrapped construction of Appendix~\ref{sec:bootstrap}. Enumeration procedures may produce repeated entries (which are included in our list-size bounds), whereas bad-seed lists are sorted and deduplicated before selection or ranking.

\subsection{Finding roots of difference polynomials}

To compute $\Bad_1(x)$, we need to find all roots in $\mathbb F_Q$ of nonzero difference polynomials $P_x-P_{x'}$. Since $\mathbb F_Q$ has characteristic two, we can find these roots deterministically using Berlekamp's trace method~\cite{Ber70}. Let $\Tr(z)=\sum_{j=0}^{b-1}z^{2^j}$ be the absolute trace from $\mathbb F_Q$ to $\mathbb F_2$, and let $\beta_i=U^{i-1}$ for $i\in[b]$ be the polynomial basis of $\mathbb F_Q$ over $\mathbb F_2$. Note that $\Tr$ is $\mathbb F_2$-linear and not identically zero (since it has degree $2^{b-1}<Q$).

\begin{procedure}[Finding roots]
\label{proc:roots}
Given a nonzero polynomial $R\in\mathbb F_Q[Z]$ of degree $d$:
\begin{enumerate}[label=\arabic*.,leftmargin=*]
  \item If $d=0$, return $\emptyset$.

  \item Compute $S=Z^Q\bmod R$ by $b$ successive squarings modulo $R$, and let $G=\gcd(R,S-Z)$ (made monic). If $\deg G=0$, return $\emptyset$.

  \item Initialize $\mathcal F=\{G\}$. For $i=1,\ldots,b$, compute $H_i=\sum_{j=0}^{b-1}(\beta_iZ)^{2^j}\bmod G$, and replace each $F\in\mathcal F$ with $\deg F>1$ by the nonconstant polynomials among $\gcd(F,H_i)$ and $\gcd(F,H_i+1)$.

  \item Return $\{\rho\in\mathbb F_Q:Z-\rho\in\mathcal F\}$.
\end{enumerate}
\end{procedure}

\begin{lemma}[Deterministic root finding]
\label{lem:roots}
Procedure~\ref{proc:roots} returns the set of all roots of $R$ in $\mathbb F_Q$ using $O(1+b^2d^2)$ operations in $\mathbb F_Q$, or $O(b^2+b^4d^2)$ bit operations.
\end{lemma}

\begin{proof}
The case $d=0$ is trivial. For $d\ge 1$, since $Z^Q-Z=\prod_{\rho\in\mathbb F_Q}(Z-\rho)$ is square-free, $G=\gcd(R,Z^Q-Z)=\gcd(R,S-Z)$ is the product of $Z-\rho$ over all distinct roots $\rho\in\mathbb F_Q$ of $R$. For any root $\rho$ of $G$, we have $H_i(\rho)=\Tr(\beta_i\rho)\in\{0,1\}$, so taking $\gcd(F,H_i)$ and $\gcd(F,H_i+1)$ splits each factor $F\in\mathcal F$ according to the value of $\Tr(\beta_i\rho)$ on its roots.

If two distinct roots $\rho\ne\rho'$ remained in the same factor after all $b$ iterations, we would have $\Tr(\beta_i(\rho-\rho'))=0$ for all $i\in[b]$, and hence $\Tr(z(\rho-\rho'))=0$ for all $z\in\mathbb F_Q$ by $\mathbb F_2$-linearity. Since multiplication by $\rho-\rho'\ne 0$ is a bijection on $\mathbb F_Q$, this would imply $\Tr\equiv 0$, a contradiction. Thus $\mathcal F$ consists of linear factors $Z-\rho$ after Step~3.

For the running time, Step~2 takes $O(bd^2)$ field operations. Letting $e=\deg G\le d$, computing $H_i$ in iteration $i$ takes $O(be^2)$ field operations. For each $F\in\mathcal F$, reducing $H_i$ modulo $F$ takes $O(e\deg F)$ field operations, after which the two gcd computations take $O((\deg F)^2)$ field operations. Their total cost is therefore $O\left(\sum_{F\in\mathcal F}\bigl(e\deg F+(\deg F)^2\bigr)\right)=O(e^2)$,
since $\sum_{F\in\mathcal F}\deg F=e$.  Summing over $b$ iterations gives $O(bd^2+b^2e^2)=O(b^2d^2)$ field operations, or $O(b^4d^2)$ bit operations.
\end{proof}

\subsection{Computing seeds and ranks}
\label{sec:computing-tags}

Both selecting the $v$-th good seed at the encoder and finding the rank of a good seed at the decoder can be done directly from a sorted list of bad seeds.

\begin{procedure}[Select and Rank]
\label{proc:select}
Given a list $\mathcal B=(\beta_1,\ldots,\beta_s)$ of distinct elements of $\mathbb F_Q$ sorted by $\iota$ (with $s=0$ allowed), let $\mathcal G=\mathbb F_Q\setminus\{\beta_1,\ldots,\beta_s\}$.
\begin{itemize}[nosep,leftmargin=*]
  \item $\textsc{Select}(v)$ (for $1\le v\le|\mathcal G|$): Set $c=v-1$. For $i=1,\ldots,s$, if $\iota(\beta_i)\le c$ then increment $c\leftarrow c+1$, else stop the loop. Return $\iota^{-1}(c)$.

  \item $\textsc{Rank}(a)$: Use binary search to test whether $a\in\mathcal B$. If so, return \emph{bad}; otherwise, return $\iota(a)+1-\bigl|\{i:\iota(\beta_i)<\iota(a)\}\bigr|$ using the count from the binary search.
\end{itemize}
\end{procedure}

\begin{lemma}[Selection and ranking]
\label{lem:select}
$\textsc{Select}(v)$ returns the $v$-th smallest element of $\mathcal G$ under $\iota$, and $\textsc{Rank}(a)$ returns the rank of $a$ in $\mathcal G$ (or \emph{bad} if $a\notin\mathcal G$), using $O(s+1)$ and $O(\log(s+2))$ comparisons, respectively.
\end{lemma}

\begin{proof}
If $\textsc{Select}(v)$ processes $j$ elements of $\mathcal B$ before stopping, it terminates with $c=v-1+j\le|\mathcal G|-1+s=Q-1$. Exactly $j$ bad elements have index at most $c$ (and all have index strictly less than $c$), so $\iota^{-1}(c)\in\mathcal G$ and there are $(c+1)-j=v$ elements of $\mathcal G$ with index at most $c$. Similarly, for $a\in\mathcal G$, $\textsc{Rank}(a)$ counts all elements with index at most $\iota(a)$ and subtracts the bad ones.
\end{proof}

Recall that for a competitor set $\Gamma(x)$, the two bad-seed sets are
\[
  \Bad_1(x)=\bigcup_{x'\in\Gamma(x)}\{\rho\in\mathbb F_Q:P_x(\rho)=P_{x'}(\rho)\}
  \quad\text{and}\quad
  \Bad_2(x,a)=\bigcup_{x'\in\Gamma(x)}\delta_{x'}^{-1}\ker\pi,
\]
where $\delta_{x'}=P_x(a)-P_{x'}(a)\ne 0$ for $a\in\Good_1(x)$, and $\ker\pi=\varphi(\{0\}^r\times\{0,1\}^{b-r})=\operatorname{span}_{\mathbb F_2}\{U^r,\ldots,U^{b-1}\}$.

\begin{lemma}[Computing seeds and ranks]
\label{lem:bad-sets}\label{lem:comp-seeds}
Let $b,r$ satisfy \eqref{eq:seed-parameters}, let $k=\lceil m/b\rceil$, and let $\Lambda$ be a list of $L'$ entries containing the exact elements of $\Gamma(x)$ (possibly with repetitions). Then:
\begin{enumerate}[label=(\alph*),leftmargin=*]
  \item The sorted, deduplicated list $\Bad_1(x)$ can be computed using $O(L'(k+m^2))$ field operations and sorting at most $s_1\le L'(k-1)$ elements of $\mathbb F_Q$.

  \item Given $a\in\Good_1(x)$, both $P_x(a)$ and the sorted, deduplicated list $\Bad_2(x,a)$ can be computed using $O\bigl(k+L'(k+2^{b-r})\bigr)$ field operations and sorting at most $s_2\le L'2^{b-r}$ elements of $\mathbb F_Q$.

  \item Once these lists are computed, selecting both seeds takes $O(s_1+s_2+2)$ comparisons, and recovering their ranks takes $O(\log(s_1+2)+\log(s_2+2))$ comparisons.
\end{enumerate}
In total, writing $s=s_1+s_2$, computing $\Tag(x,\sigma,\tau)$ or recovering $(\sigma,\tau)$ from $x$ and its tag takes
\[
  O\bigl(m+b+(L'+1)(m^2+2^{b-r})b^2+(s+1)b\log(s+2)\bigr)
\]
bit operations. For the parameters of Section~\ref{sec:dl-codes}, this is $O_t(n^{2t+2}\log^2 n)$ bit operations.
\end{lemma}

\begin{proof}
Padding $x$ and reading the coefficients of $P_x$ takes $O(kb)=O(m+b)$ bit operations. For each $x'\in\Lambda$, forming $R=P_x-P_{x'}$ takes $O(k)$ field operations, and since $\deg R\le k-1$ and $b(k-1)<m$, Procedure~\ref{proc:roots} finds the at most $k-1$ roots of $R$ in $O(1+b^2(k-1)^2)=O(1+m^2)$ field operations. Sorting and deduplicating the $s_1\le L'(k-1)$ roots gives $\Bad_1(x)$.

For (b), evaluating $P_x(a)$ and each $\delta_{x'}=P_x(a)-P_{x'}(a)$ by Horner's rule takes $O((L'+1)k)$ field operations. Since $a\in\Good_1(x)$, each $\delta_{x'}$ is nonzero; after computing $\delta_{x'}^{-1}$ and the basis $\delta_{x'}^{-1}U^r,\ldots,\delta_{x'}^{-1}U^{b-1}$, we can list the $2^{b-r}$ elements of $\delta_{x'}^{-1}\ker\pi$ in Gray-code order using one field addition per element. Sorting and deduplicating the resulting $s_2\le L'2^{b-r}$ elements gives $\Bad_2(x,a)$. Part (c) follows from Lemma~\ref{lem:select}.

Finally, multiplying field operations by $O(b^2)$ and comparisons by $O(b)$ gives the bit-operation bound. For the parameters of Section~\ref{sec:dl-codes}, we have $L'=O_t(n^{2t})$, $b=O_t(\log n)$, $m<n$, and $2^{b-r}\le 2n$, so $s=O_t(n^{2t+1})$, the field-arithmetic cost is $O_t(n^{2t+2}\log^2 n)$, and the sorting cost is $O_t(n^{2t+1}\log^2 n)$.
\end{proof}

\subsection{Enumerating competitors and preimages}
\label{subsec:dl-enumeration}

We now record the enumeration procedures underlying the list-size and running-time bounds in Section~\ref{subsec:dl-competitors}. The payload competitors, suffix neighborhoods, and channel preimages are all obtained by enumerating DL balls and filtering by length (with repetitions included in all bounds below).

\begin{lemma}[Enumerating DL balls]
\label{lem:dl-ball}
For every $u\in\{0,1\}^{\ell}$ and integer $q\ge 0$, the ball $B_q(u)$ can be listed deterministically (possibly with repetitions) in $O\bigl(M(\ell,q)(\ell+q+1)\bigr)$ bit operations with at most
\[
  M(\ell,q):=(q+1)\bigl(5(\ell+q)+2\bigr)^q
\]
entries.
\end{lemma}

\begin{proof}
We traverse the tree of depth $q$ rooted at $u$, where the children of each node are obtained by applying one DL operation to the string at that node, and output the string at every visited node. Every string in $B_q(u)$ has length at most $\ell+q$, and a string of length $a$ admits $2(a+1)$ insertions, $a$ deletions, $a$ substitutions, and at most $a$ adjacent transpositions, for a total of at most $5a+2\le 5(\ell+q)+2$ operations. Thus the tree has at most $\sum_{j=0}^{q}(5(\ell+q)+2)^j\le M(\ell,q)$ nodes, and generating each child from its parent takes $O(\ell+q+1)$ bit operations.
\end{proof}

For fixed $t\ge 1$ and $n=m+N\ge 2t+1$, recall that $D_t^{\mathrm{DL}}=(2t+1)10^{2t}$. Applying Lemma~\ref{lem:dl-ball} gives the three required lists:
\begin{itemize}[nosep,leftmargin=*]
  \item \emph{Payload competitors:} By Corollary~\ref{cor:dl-competitors}, $\Gamma_{\mathrm{DL}}(x)=(B_{2t}(x)\cap\{0,1\}^{m})\setminus\{x\}$. Listing $B_{2t}(x)$ and keeping strings of length $m$ different from $x$ gives at most $M(m,2t)\le D_t^{\mathrm{DL}}n^{2t}$ entries in $O_t(n^{2t+1})$ bit operations (since $5(m+2t)+2\le 10n$).

  \item \emph{Suffix neighborhoods:} For $s\in\{0,1\}^{N}$, listing $B_{2t}(s)$ and keeping strings of length $N$ gives the $\mathcal R_N$-neighborhood $B_{2t}(s)\cap\{0,1\}^{N}$ in $O_t(N^{2t+1})$ bit operations, with size bounded by $D_N\le D_t^{\mathrm{DL}}N^{2t}$.

  \item \emph{Decoder preimages:} For any received word $y$ with $|y|\in[n-t,n+t]$, listing $B_t(y)$ and keeping strings of length $n$ gives $\mathcal E_n^{-1}(y)=B_t(y)\cap\{0,1\}^{n}$ with $O_t(n^t)$ entries in $O_t(n^{t+1})$ bit operations (while inputs with $|y|\notin[n-t,n+t]$ have $\mathcal E_n^{-1}(y)=\emptyset$).
\end{itemize}
Omitting adjacent transpositions from Lemma~\ref{lem:dl-ball} gives the analogous enumeration for ordinary edit balls $\mathcal E_q(u)$, which we use in Appendix~\ref{subsec:edit-decoding}.

\subsection{Encoding and decoding complexity for DL errors}

Combining Lemma~\ref{lem:bad-sets} with the competitor and preimage enumerations of Section~\ref{subsec:dl-competitors} yields the complexity bounds stated in Theorem~\ref{thm:dl-codes}.

\begin{theorem}[Encoding and decoding complexity]
\label{thm:complexity}
Fix $t\ge 1$ and the parameters of Section~\ref{sec:dl-codes}. After precomputation, the encoder \eqref{eq:protected-tag-encoder} and the generic decoder (Procedure~\ref{proc:generic-dec}) run in the following numbers of bit operations:
\begin{enumerate}[label=(\alph*),leftmargin=*]
  \item encoding: $O_t(n^{2t+2}\log^2 n)+O(\tau_A)$;
  \item payload recovery: $O_t(n^{t+1}\log n)+O_t(n^t)\tau_A$;
  \item full-message recovery: $O_t(n^{2t+2}\log^2 n)+O_t(n^t)\tau_A$.
\end{enumerate}
\end{theorem}

\begin{proof}
Encoding lists $\Gamma_{\mathrm{DL}}(x)$ in $O_t(n^{2t+1})$ bit operations, computes the tag in $O_t(n^{2t+2}\log^2 n)$ bit operations by Lemma~\ref{lem:bad-sets}, and applies $A$ in time $\tau_A$. For decoding, after rejecting inputs with $|y|\notin[n-t,n+t]$, Step~1 of Procedure~\ref{proc:generic-dec} lists $\mathcal E_n^{-1}(y)=B_t(y)\cap\{0,1\}^n$ ($O_t(n^t)$ entries in $O_t(n^{t+1})$ time) and tests each candidate using time $\tau_A$ for $A^{-1}$ and $O(kb^2)=O_t(n\log n)$ bit operations for Horner's rule, giving (b). Once a candidate passes Step~1, Steps~2 and~3 compute $\Bad_1$ and $\Bad_2$ at most once, which adds $O_t(n^{2t+2}\log^2 n)$ bit operations by Lemma~\ref{lem:bad-sets}.
\end{proof}

\subsection{Reducing tag-code queries for edits}
\label{subsec:edit-decoding}

When the channel introduces only insertions, deletions, and substitutions (i.e., $\dL(c,y)\le t$), we can decode the exact same DL code of Section~\ref{sec:dl-codes} while reducing the number of tag-code queries from $O_t(n^t)$ to $O_t(\log^t n)$. The key observation is that under $t$ ordinary edits, the boundary between $x$ and $s$ shifts by at most $t$ positions in $y$, allowing us to recover candidate tags from the suffix of $y$ before enumerating payload candidates.

\begin{lemma}[Locating the boundary under edits]
\label{lem:edit-boundary}
If $c=x\Vert s$ with $|x|=m$ and $\dL(c,y)\le t$, then there is an index $j\in[\max\{0,m-t\},\min\{m+t,|y|\}]$ such that $\dL(x,y_{[1,j]})\le t$ and $\dL(s,y_{[j+1,|y|]})\le t$.
\end{lemma}

\begin{proof}
Fix an optimal edit alignment between $c=x\Vert s$ and $y$ of cost at most $t$, and cut it immediately after the column containing the last bit of $x$. If $j$ bits of $y$ lie to the left of the cut and the two halves of the alignment have costs $t_1$ and $t_2$, then $t_1+t_2\le t$, $\dL(x,y_{[1,j]})\le t_1$, $\dL(s,y_{[j+1,|y|]})\le t_2$, and $|j-m|\le t_1\le t$.
\end{proof}

Below, $\mathcal E_t(u)=\{v:\dL(u,v)\le t\}$ denotes the edit ball of radius $t$, which can be listed with $O_t((|u|+1)^t)$ entries in $O_t((|u|+1)^{t+1})$ time by omitting transpositions from the enumeration of Section~\ref{subsec:dl-competitors}.

\begin{procedure}[Decoder for ordinary edits]
\label{proc:dec}
Given a string $y$ with $|y|\in[n-t,n+t]$:
\begin{enumerate}[label=\arabic*.,leftmargin=*]
  \item \emph{Candidate tags.} Initialize $\mathcal T=\emptyset$. For each $j\in[\max\{0,m-t\},\min\{m+t,|y|\}]$, scan $u\in\mathcal E_t(y_{[j+1,|y|]})\cap\{0,1\}^N$; at the first $u\in A(\{0,1\}^h)$, add $(j,A^{-1}(u),u)$ to $\mathcal T$ and move to the next $j$.

  \item \emph{Candidate payloads.} For each $(j,T,u)\in\mathcal T$, write $T=\varphi^{-1}(a)\Vert\varphi^{-1}(\lambda)\Vert\eta$ and scan $x'\in\mathcal E_t(y_{[1,j]})\cap\{0,1\}^m$. If a candidate satisfies both $\pi(\lambda P_{x'}(a))=\eta$ and $\dL(x'\Vert u,y)\le t$, retain $(x',a,\lambda)$ and proceed to Step~3. (To recover only the payload, output $x'$ instead.) Return \emph{failure} only after all triples in $\mathcal T$ and their associated payload candidates have been exhausted.

  \item \emph{Seed ranks.} Run Steps~2 and~3 of Procedure~\ref{proc:generic-dec} on $(x',a,\lambda)$ using $\Gamma_{\mathrm{DL}}(x')$.
\end{enumerate}
\end{procedure}

\begin{lemma}[Correctness and cost of the edit decoder]
\label{lem:dec-correct}
For any codeword $c=\Enc(x,\sigma,\tau)$ of the DL code and any $y$ with $\dL(c,y)\le t$, Procedure~\ref{proc:dec} outputs $(x,\sigma,\tau)$ using $O_t(\log^t n)$ tag-code queries, $O_t(n^{t+1}\log n)+O_t(\log^t n)\tau_A$ bit operations for payload recovery, and $O_t(n^{2t+2}\log^2 n)+O_t(\log^t n)\tau_A$ bit operations for full-message recovery.
\end{lemma}

\begin{proof}
Write $c=x\Vert A(T)$. By Lemma~\ref{lem:edit-boundary}, there is $j^\ast$ with $A(T)\in\mathcal E_t(y_{[j^\ast+1,|y|]})$ and $x\in\mathcal E_t(y_{[1,j^\ast]})$. For any fixed $j$, two tag-codewords $u,u'\in\mathcal E_t(y_{[j+1,|y|]})\cap A(\{0,1\}^h)$ satisfy $d_{\mathrm{DL}}(u,u')\le\dL(u,u')\le 2t$ and hence $u=u'$. Thus Step~1 adds $(j^\ast,T,A(T))$ to $\mathcal T$, and $x$ passes Step~2. Conversely, any $(x',u)$ passing Step~2 satisfies $d_{\mathrm{DL}}(x'\Vert u,y)\le\dL(x'\Vert u,y)\le t$ and passes Step~1 of Procedure~\ref{proc:generic-dec}, so Lemma~\ref{lem:generic-dec} gives $x'\Vert u=c$ and correctness of Step~3.

For the running time, Step~1 checks at most $2t+1$ values of $j$; for each $j$, $|y_{[j+1,|y|]}|\le N+2t=O_t(\log n)$, so $\mathcal E_t(y_{[j+1,|y|]})$ has $O_t(\log^t n)$ entries and Step~1 makes $O_t(\log^t n)$ queries to $A$. For each of the at most $2t+1$ triples in $\mathcal T$, Step~2 enumerates $O_t(n^t)$ candidates $x'$, spending $O_t(n\log n)$ bit operations for Horner's rule and $O_t(n)$ bit operations to test $\dL(x'\Vert u,y)\le t$ via banded dynamic programming~\cite{Ukk85}. Finally, Step~3 takes $O_t(n^{2t+2}\log^2 n)$ bit operations by Lemma~\ref{lem:bad-sets}.
\end{proof}

\begin{table}[t]
\centering
\small
\caption{Summary of bit-operation costs for fixed $t$ and the parameters of Section~\ref{sec:dl-codes}.}
\label{tab:costs}
\begin{tabular}{@{}p{0.62\linewidth}p{0.33\linewidth}@{}}
\toprule
Operation & Bit operations \\
\midrule
Enumerate $\Gamma_{\mathrm{DL}}(x)$ & $O_t(n^{2t+1})$ \\
Find roots of difference polynomials (Procedure~\ref{proc:roots}) & $O_t(n^{2t+2}\log^2 n)$ \\
Evaluate differences and enumerate scaled copies $\delta_{x'}^{-1}\ker\pi$ & $O_t(n^{2t+1}\log^2 n)$ \\
Sort and deduplicate $\Bad_1(x)$ and $\Bad_2(x,a)$ & $O_t(n^{2t+1}\log^2 n)$ \\
Select seeds or recover ranks (Procedure~\ref{proc:select}) & $O_t(n^{2t+1}\log n)$ \\
\midrule
Enumerate preimages $\mathcal E_n^{-1}(y)$ & $O_t(n^{t+1})$ \\
Evaluate candidate hashes (Horner's rule) & $O_t(n^{t+1}\log n)$ \\
Tag-code queries: generic DL decoder (Procedure~\ref{proc:generic-dec}) & $O_t(n^t)\tau_A$ \\
Tag-code queries: edit decoder (Procedure~\ref{proc:dec}) & $O_t(\log^t n)\tau_A$ \\
\midrule
Precompute irreducible polynomial $f_b$~\cite{Sho90} & $\polylog(n)$ \\
Precompute tag code $A$: greedy / bootstrapped & $n^{6t+2+o(1)}$ / $\polylog(n)$ \\
\bottomrule
\end{tabular}
\end{table}

\section{Bootstrapping the tag code}
\label{sec:bootstrap}

The greedy tag-code construction in Section~\ref{sec:dl-codes} requires $n^{6t+2+o(1)}$ precomputation time. Since the tag itself has length $h=O_t(\log n)$ and the DL channel is self-splittable, we can reduce the precomputation time to $\operatorname{polylog}(n)$ by using our own greedy code construction at block length $N_h=O_t(\log n)$ as the tag code. We first state this one-step bootstrapping reduction for a general code family and then specialize it to the DL code family.

\subsection{A general bootstrapping lemma}

Given a family of codes of block length $\ell$ and message length $M(\ell)$, we obtain a code for $h$-bit tags by choosing the smallest $\ell>h$ with $M(\ell)\ge h$ and padding the $h$-bit tag with trailing zeroes. To invert this map and test membership in its image, we run the decoder of the length-$\ell$ code, check that the recovered message ends with zeroes, and re-encode.

\begin{lemma}[Bootstrapping]
\label{lem:bootstrap}
Let $\mathcal E=(\mathcal E_\ell)_{\ell\ge 1}$ be a channel with $u\in\mathcal E_\ell(u)$ for all $u\in\{0,1\}^{\ell}$. Suppose that for every $\ell\ge\ell_0$, there is a code $\mathcal C_\ell\subseteq\{0,1\}^{\ell}$ correcting $\mathcal E_\ell$, an injective encoder $\Enc_\ell:\{0,1\}^{M(\ell)}\to\mathcal C_\ell$, and a decoder $\Dec_\ell$ that halts on every input and returns $\mu$ on every valid channel output $y\in\mathcal E_\ell(\Enc_\ell(\mu))$. Assume that given $\ell$, one can compute $M(\ell)$ in time $\theta(\ell)$, perform the precomputation for $\mathcal C_\ell$ in time $P(\ell)$, and then run $\Enc_\ell$ and $\Dec_\ell$ in time $\rm Time(\ell)$, where $\theta,P,\rm Time,R$ are nondecreasing functions satisfying $\ell-M(\ell)\le R(\ell)\le\ell/2$ for all $\ell\ge\ell_0$.

Then for every $h\ge\ell_0$, there is an integer $N_h$ with $h<N_h\le 2h$ and $N_h-h\le 1+R(2h)$, together with an injective map $A_h:\{0,1\}^{h}\to\{0,1\}^{N_h}$ whose images are pairwise non-confusable under $\mathcal E_{N_h}$. Finding $N_h$ and precomputing $A_h$ takes $O\bigl(h\theta(2h)+P(2h)\bigr)$ time, after which evaluating $A_h$, testing membership in $A_h(\{0,1\}^{h})$, and computing $A_h^{-1}$ each take $O({\rm Time}(2h)+h)$ time.
\end{lemma}

\begin{proof}
Let $N_h=\min\{\ell\in\{h+1,\ldots,2h\}:M(\ell)\ge h\}$, which is well-defined since $M(2h)\ge 2h-R(2h)\ge h$. If $N_h=h+1$, then $N_h-h=1$; otherwise $M(N_h-1)<h$, so $N_h-h<1+(N_h-1-M(N_h-1))\le 1+R(N_h-1)\le 1+R(2h)$.

For $T\in\{0,1\}^{h}$, define $A_h(T)=\Enc_{N_h}\bigl(T\Vert 0^{M(N_h)-h}\bigr)$. Since $\Enc_{N_h}$ is injective and its image corrects $\mathcal E_{N_h}$, $A_h$ is injective and distinct images of $A_h$ are non-confusable under $\mathcal E_{N_h}$. To test whether $u\in\{0,1\}^{N_h}$ lies in $A_h(\{0,1\}^{h})$ and compute $A_h^{-1}(u)$, we run $\Dec_{N_h}(u)$, check that it returns a message $\mu\in\{0,1\}^{M(N_h)}$ ending with $M(N_h)-h$ zeroes and satisfying $\Enc_{N_h}(\mu)=u$, and if so output $\mu_{[1,h]}$. If $u=A_h(T)$, then $u\in\mathcal E_{N_h}(u)$, so $\Dec_{N_h}(u)$ returns $T\Vert 0^{M(N_h)-h}$ and the test accepts with output $T$; conversely, if the re-encoding check $\Enc_{N_h}(\mu)=u$ passes, then $u=A_h(\mu_{[1,h]})$ regardless of how $\Dec_{N_h}$ behaves on invalid inputs.

Finally, finding $N_h$ by testing $\ell=h+1,\ldots,2h$ takes $O(h\theta(2h))$ time, precomputing $\mathcal C_{N_h}$ takes at most $P(2h)$ time, and each operation of $A_h$ makes at most one call to $\Enc_{N_h}$ and one call to $\Dec_{N_h}$ in time $O({\rm Time}(2h)+h)$.
\end{proof}

\subsection{Application to DL codes}
\label{subsec:bootstrap-dl}

Fix $t\ge 1$, and let $(\mathcal C_\ell)_{\ell\ge\ell_0}$ be the greedy DL code family of Theorem~\ref{thm:dl-codes}(b) at block length $\ell$, decoded with Procedure~\ref{proc:generic-dec}. At length $\ell$, the code uses parameters $L_\ell=D_t^{\mathrm{DL}}\ell^{2t}$, $b_\ell=\lceil\log(2L_\ell\ell)\rceil$, $r_\ell=\lceil\log(2L_\ell)\rceil$, $h_\ell=2b_\ell+r_\ell$, and suffix length $N_{\mathrm g}(\ell)=h_\ell+\lceil 2t\log(20h_\ell)+\log(2t+1)\rceil$. Its message length $M(\ell)=\ell-N_{\mathrm g}(\ell)+2b_\ell-2$ is a closed-form function computable in $\theta(\ell)=\operatorname{polylog}(\ell)$ time, with redundancy $\ell-M(\ell)\le R(\ell):=2t\log\ell+2t\log\log\ell+c_t$. Enlarge $\ell_0$ so that the greedy construction is valid and $R(\ell)\le\ell/2$ for every $\ell\ge\ell_0$. For sufficiently large $n$, the outer tag length satisfies $h\ge\ell_0$. By Theorem~\ref{thm:dl-codes}(b) and Theorem~\ref{thm:complexity}, we may take the nondecreasing bounds
\[
  P(\ell)=C_t\ell^{6t+2}(\log\ell)^{4t+1},
  \qquad
  {\rm Time}(\ell)=C_t\ell^{2t+2}\log^2\ell
\]
for a sufficiently large constant $C_t$. In particular, $P(\ell)=\ell^{6t+2+o(1)}$.

Applying Lemma D.1 to this family for an outer tag of length $h=(6t+2)\log n+O_t(1)$ yields a tag code $A_h:\{0,1\}^h\to\{0,1\}^{N_h}$ satisfying~\eqref{eq:tag-code-separation} for $N=N_h$. Its overhead satisfies
\[
  N_h-h\le 1+R(2h)=2t\log h+2t\log\log h+O_t(1)=2t\log\log n+O_t(\log\log\log n).
\]
Because $h=O_t(\log n)$, the precomputation time for $A_h$ (whose own inner tag code is a greedy table of length $O_t(\log h)=O_t(\log\log n)$) is $O\bigl(h\theta(2h)+P(2h)\bigr)=(\log n)^{6t+2+o(1)}=\operatorname{polylog}(n)$, and each operation of $A_h$ takes time $\tau_A=O({\rm Time}(2h)+h)=O_t(\log^{2t+2} n\cdot\log^2\log n)=\operatorname{polylog}(n)$.  Substituting $A_h$ and $\tau_A=\operatorname{polylog}(n)$ into Lemma~\ref{lem:tag-to-code} and Theorem~\ref{thm:complexity} gives the bootstrapped code of Theorem~\ref{thm:dl-codes}(a).

\section*{AI methodology statement}
We used generative AI tools (OpenAI GPT-6 Astra and Google Gemini Flash 3.8) to assist in the research and preparation of this manuscript. Specifically, the tools were used to explore and streamline the boundary-projection and competitor-ball framework in Section~\ref{sec:dl-codes}, as well as for language polishing, clarity improvements, and TikZ typesetting across the manuscript. The authors conceived the core seed recycling mechanism, independently checked and established the correctness and rigor of all mathematical statements, proofs, and analyses, and verified the accuracy and originality of all content, including references. The authors take full responsibility for the entire submission.

\bibliographystyle{alpha}
\bibliography{refs}

@article{Lev66,
  author  = {V. I. Levenshtein},
  title   = {Binary codes capable of correcting deletions, insertions, and reversals},
  journal = {Soviet Physics Doklady},
  volume  = {10},
  number  = {8},
  pages   = {707--710},
  year    = {1966}
}

@article{VT65,
  author  = {Varshamov, R. R. and Tenengol'ts, G. M.},
  title   = {Code Correcting Single Asymmetric Errors},
  journal = {Avtomatika i Telemekhanika},
  year    = {1965},
  volume  = {26},
  number  = {2},
  pages   = {288--292},
  note    = {In Russian},
  url     = {http://mi.mathnet.ru/at11293}
}

@inproceedings{Lev02,
  author    = {V. I. Levenshtein},
  title     = {Bounds for deletion/insertion correcting codes},
  booktitle = {Proc. IEEE International Symposium on Information Theory (ISIT)},
  year      = {2002}
}

@article{ullman1966near,
  title={Near-optimal, single-synchronization-error-correcting code},
  author={Ullman, J},
  journal={IEEE Transactions on Information Theory},
  volume={12},
  number={4},
  pages={418--424},
  year={1966},
  publisher={IEEE}
}

@article{ABGHK24,
  author  = {N. Alon and G. Bourla and B. Graham and X. He and N. Kravitz},
  title   = {Logarithmically larger deletion codes of all distances},
  journal = {IEEE Transactions on Information Theory},
  volume  = {70},
  number  = {1},
  pages   = {125--130},
  year    = {2024}
}

@misc{EnGad26,
  author = {E. En Gad},
  title  = {Polynomially larger deletion codes by linear hashing of substring counts},
  year   = {2026},
  note   = {arXiv:2609.19493v2}
}

@article{BGZ18,
  author  = {J. Brakensiek and V. Guruswami and S. Zbarsky},
  title   = {Efficient low-redundancy codes for correcting multiple deletions},
  journal = {IEEE Transactions on Information Theory},
  volume  = {64},
  number  = {5},
  pages   = {3403--3410},
  year    = {2018},
  note    = {Preliminary version in SODA 2016}
}

@article{GS19,
  author  = {R. Gabrys and F. Sala},
  title   = {Codes correcting two deletions},
  journal = {IEEE Transactions on Information Theory},
  volume  = {65},
  number  = {2},
  pages   = {965--974},
  year    = {2019}
}

@article{SRB20,
  author  = {J. Sima and N. Raviv and J. Bruck},
  title   = {Two deletion correcting codes from indicator vectors},
  journal = {IEEE Transactions on Information Theory},
  volume  = {66},
  number  = {4},
  pages   = {2375--2391},
  year    = {2020}
}

@article{GH21,
  author  = {V. Guruswami and J. H{\aa}stad},
  title   = {Explicit two-deletion codes with redundancy matching the existential bound},
  journal = {IEEE Transactions on Information Theory},
  volume  = {67},
  number  = {10},
  pages   = {6384--6394},
  year    = {2021}
}

@inproceedings{CJLW18,
  author    = {K. Cheng and Z. Jin and X. Li and K. Wu},
  title     = {Deterministic document exchange protocols, and almost optimal binary codes for edit errors},
  booktitle = {Proc. 59th IEEE Symposium on Foundations of Computer Science (FOCS)},
  pages     = {200--211},
  year      = {2018}
}

@inproceedings{Hae19,
  author    = {B. Haeupler},
  title     = {Optimal document exchange and new codes for insertions and deletions},
  booktitle = {Proc. 60th IEEE Symposium on Foundations of Computer Science (FOCS)},
  pages     = {334--347},
  year      = {2019}
}

@inproceedings{HS17,
  author    = {B. Haeupler and A. Shahrasbi},
  title     = {Synchronization strings: codes for insertions and deletions approaching the {S}ingleton bound},
  booktitle = {Proc. 49th ACM Symposium on Theory of Computing (STOC)},
  pages     = {33--46},
  year      = {2017}
}

@inproceedings{GHS20,
  author    = {V. Guruswami and B. Haeupler and A. Shahrasbi},
  title     = {Optimally resilient codes for list-decoding from insertions and deletions},
  booktitle = {Proc. 52nd ACM Symposium on Theory of Computing (STOC)},
  pages     = {524--537},
  year      = {2020}
}

@inproceedings{GHL21,
  author    = {V. Guruswami and X. He and R. Li},
  title     = {The zero-rate threshold for adversarial bit-deletions is less than 1/2},
  booktitle = {Proc. 62nd IEEE Symposium on Foundations of Computer Science (FOCS)},
  year      = {2021}
}

@article{SZ99,
  author  = {L. J. Schulman and D. Zuckerman},
  title   = {Asymptotically good codes correcting insertions, deletions, and transpositions},
  journal = {IEEE Transactions on Information Theory},
  volume  = {45},
  number  = {7},
  pages   = {2552--2557},
  year    = {1999}
}

@article{SB21,
  author  = {J. Sima and J. Bruck},
  title   = {On optimal $k$-deletion correcting codes},
  journal = {IEEE Transactions on Information Theory},
  volume  = {67},
  number  = {6},
  pages   = {3360--3375},
  year    = {2021}
}

@inproceedings{SGB20syndrome,
  author    = {J. Sima and R. Gabrys and J. Bruck},
  title     = {Syndrome compression for optimal redundancy codes},
  booktitle = {Proc. IEEE International Symposium on Information Theory (ISIT)},
  pages     = {751--756},
  year      = {2020}
}

@inproceedings{SGB20opt,
  author    = {J. Sima and R. Gabrys and J. Bruck},
  title     = {Optimal systematic $t$-deletion correcting codes},
  booktitle = {Proc. IEEE International Symposium on Information Theory (ISIT)},
  pages     = {769--774},
  year      = {2020}
}

@inproceedings{SPCH21,
  author    = {W. Song and N. Polyanskii and K. Cai and X. He},
  title     = {On multiple-deletion multiple-substitution correcting codes},
  booktitle = {Proc. IEEE International Symposium on Information Theory (ISIT)},
  pages     = {2655--2660},
  year      = {2021}
}

@article{SPCH22,
  author  = {W. Song and N. Polyanskii and K. Cai and X. He},
  title   = {Systematic codes correcting multiple-deletion and multiple-substitution errors},
  journal = {IEEE Transactions on Information Theory},
  volume  = {68},
  number  = {10},
  pages   = {6402--6416},
  year    = {2022}
}

@misc{LGF25,
  author = {Y. Li and R. Gabrys and F. Farnoud},
  title  = {Constructing low-redundancy codes via distributed graph coloring},
  year={2025},
  note={arXiv:2512.04197v1}
}

@article{SWWY23,
  author  = {I. Smagloy and L. Welter and A. Wachter-Zeh and E. Yaakobi},
  title   = {Single-deletion single-substitution correcting codes},
  journal = {IEEE Transactions on Information Theory},
  volume  = {69},
  number  = {12},
  pages   = {7659--7671},
  year    = {2023}
}

@misc{SG24,
  author = {Y. Sun and G. Ge},
  title  = {Binary codes for correcting two edits},
  year   = {2024},
  note   = {arXiv:2403.11766}
}

@misc{PZ24,
  author = {Y. Pi and Z. Zhang},
  title  = {Two-insertion/deletion/substitution correcting codes},
  year   = {2024},
  note   = {arXiv:2401.11231}
}

@article{Lin92,
  author  = {N. Linial},
  title   = {Locality in distributed graph algorithms},
  journal = {SIAM Journal on Computing},
  volume  = {21},
  number  = {1},
  pages   = {193--201},
  year    = {1992}
}

@article{Ber70,
  author  = {E. R. Berlekamp},
  title   = {Factoring polynomials over large finite fields},
  journal = {Mathematics of Computation},
  volume  = {24},
  number  = {111},
  pages   = {713--735},
  year    = {1970}
}

@article{Sho90,
  author  = {V. Shoup},
  title   = {New algorithms for finding irreducible polynomials over finite fields},
  journal = {Mathematics of Computation},
  volume  = {54},
  number  = {189},
  pages   = {435--447},
  year    = {1990}
}

@article{Ukk85,
  author  = {E. Ukkonen},
  title   = {Algorithms for approximate string matching},
  journal = {Information and Control},
  volume  = {64},
  number  = {1--3},
  pages   = {100--118},
  year    = {1985}
}

@article{CGK12,
  author  = {G. M. Church and Y. Gao and S. Kosuri},
  title   = {Next-generation digital information storage in {DNA}},
  journal = {Science},
  volume  = {337},
  number  = {6102},
  pages   = {1628},
  year    = {2012}
}

@article{PHT08,
  author  = {S. S. P. Parkin and M. Hayashi and L. Thomas},
  title   = {Magnetic domain-wall racetrack memory},
  journal = {Science},
  volume  = {320},
  number  = {5873},
  pages   = {190--194},
  year    = {2008}
}

@article{Mit09,
  author  = {M. Mitzenmacher},
  title   = {A survey of results for deletion channels and related synchronization channels},
  journal = {Probability Surveys},
  volume  = {6},
  pages   = {1--33},
  year    = {2009}
}

@inproceedings{CPSV00,
  author    = {G. Cormode and M. Paterson and S. C. Sahinalp and U. Vishkin},
  title     = {Communication complexity of document exchange},
  booktitle = {Proc. 11th ACM-SIAM Symposium on Discrete Algorithms (SODA)},
  pages     = {197--206},
  year      = {2000}
}

@article{Dam64,
  author  = {F. J. Damerau},
  title   = {A technique for computer detection and correction of spelling errors},
  journal = {Communications of the ACM},
  volume  = {7},
  number  = {3},
  pages   = {171--176},
  year    = {1964}
}

@article{GYM18,
  author  = {R. Gabrys and E. Yaakobi and O. Milenkovic},
  title   = {Codes in the {D}amerau distance for deletion and adjacent transposition correction},
  journal = {IEEE Transactions on Information Theory},
  volume  = {64},
  number  = {4},
  year    = {2018}
}

@misc{WVT23,
  author = {S. Wang and V. K. Vu and V. Y. F. Tan},
  title  = {Codes for correcting asymmetric adjacent transpositions and deletions},
  year   = {2023},
  note   = {arXiv:2301.11680. Preliminary versions in ITW 2022 and ISIT 2023}
}

@article{Ten84,
  author  = {G. Tenengolts},
  title   = {Nonbinary codes, correcting single deletion or insertion},
  journal = {IEEE Transactions on Information Theory},
  volume  = {30},
  number  = {5},
  pages   = {766--769},
  year    = {1984}
}

@article{NCS24,
  author  = {T. T. Nguyen and K. Cai and P. H. Siegel},
  title   = {A new version of $q$-ary {V}arshamov--{T}enengolts codes with more efficient encoders: the differential {VT} codes and the differential shifted {VT} codes},
  journal = {IEEE Transactions on Information Theory},
  volume  = {70},
  number  = {10},
  year    = {2024}
}

@article{SC23,
  author  = {W. Song and K. Cai},
  title   = {Non-binary two-deletion correcting codes and burst-deletion correcting codes},
  journal = {IEEE Transactions on Information Theory},
  volume  = {69},
  number  = {10},
  pages   = {6470--6484},
  year    = {2023}
}

@misc{LTX24,
  author = {S. Liu and I. Tjuawinata and C. Xing},
  title  = {Explicit construction of $q$-ary 2-deletion correcting codes with low redundancy},
  year   = {2023},
  note   = {arXiv:2306.02868. Also in IEEE Transactions on Information Theory, 2024}
}

@inproceedings{SGB20optq,
  author    = {J. Sima and R. Gabrys and J. Bruck},
  title     = {Optimal codes for the $q$-ary deletion channel},
  booktitle = {Proc. IEEE International Symposium on Information Theory (ISIT)},
  pages     = {740--745},
  year      = {2020}
}

@article{LN16,
  author  = {T. A. Le and H. D. Nguyen},
  title   = {New multiple insertion/deletion correcting codes for non-binary alphabets},
  journal = {IEEE Transactions on Information Theory},
  volume  = {62},
  number  = {5},
  pages   = {2682--2693},
  year    = {2016}
}

@misc{STV25,
  author = {M. Schaller and B. Toesca and V. K. Vu},
  title  = {A new construction of non-binary deletion correcting codes and their decoding},
  year   = {2025},
  note   = {arXiv:2501.13534. Preliminary version in ISIT 2025}
}

@article{JMSB09,
  author  = {A. Jiang and R. Mateescu and M. Schwartz and J. Bruck},
  title   = {Rank modulation for flash memories},
  journal = {IEEE Transactions on Information Theory},
  volume  = {55},
  number  = {6},
  pages   = {2659--2673},
  year    = {2009}
}

@article{Lev92,
  author  = {V. I. Levenshtein},
  title   = {On perfect codes in deletion and insertion metric},
  journal = {Discrete Mathematics and Applications},
  volume  = {2},
  number  = {3},
  pages   = {241--258},
  year    = {1992}
}

@article{FSM13,
  author  = {F. Farnoud and V. Skachek and O. Milenkovic},
  title   = {Error-correction in flash memories via codes in the {U}lam metric},
  journal = {IEEE Transactions on Information Theory},
  volume  = {59},
  number  = {5},
  pages   = {3003--3020},
  year    = {2013}
}

@article{GYFSBD16,
  author  = {R. Gabrys and E. Yaakobi and F. Farnoud and F. Sala and J. Bruck and L. Dolecek},
  title   = {Codes correcting erasures and deletions for rank modulation},
  journal = {IEEE Transactions on Information Theory},
  volume  = {62},
  number  = {1},
  pages   = {136--150},
  year    = {2016}
}

@article{WNCV25,
  author  = {S. Wang and T. Nguyen and Y. M. Chee and V. K. Vu},
  title   = {Permutation and multi-permutation codes correcting multiple deletions},
  journal = {IEEE Transactions on Information Theory},
  volume  = {71},
  number  = {9},
  pages   = {6759--6770},
  year    = {2025}
}

@misc{WBNVW26,
  author = {B. Wu and Q. D. Bui and K. Niu and V. K. Vu and S. Wang},
  title  = {Order-optimal systematic permutation codes for correcting $t$ deletions},
  year   = {2026},
  note   = {arXiv:2609.37578}
}

@article{CB11,
  author  = {Y. Cassuto and M. Blaum},
  title   = {Codes for symbol-pair read channels},
  journal = {IEEE Transactions on Information Theory},
  volume  = {57},
  number  = {12},
  pages   = {8011--8020},
  year    = {2011}
}

@inproceedings{CV20,
  author    = {Y. M. Chee and V. K. Vu},
  title     = {Codes correcting synchronization errors for symbol-pair read channels},
  booktitle = {Proc. IEEE International Symposium on Information Theory (ISIT)},
  pages     = {746--750},
  year      = {2020}
}

@misc{XC26,
  author = {H. Xie and Z. Chen},
  title  = {First-order optimal codes for an adversarial nanopore channel},
  year   = {2026},
  note   = {arXiv:2601.21236}
}

@misc{YG26,
  author = {Z. Ye and G. Ge},
  title  = {Deletion-correcting codes for the $\ell$-symbol read channel},
  year   = {2026},
  note   = {arXiv:2606.26434}
}

@article{CKVVY18,
  author  = {Y. M. Chee and H. M. Kiah and A. Vardy and V. K. Vu and E. Yaakobi},
  title   = {Coding for racetrack memories},
  journal = {IEEE Transactions on Information Theory},
  volume  = {64},
  number  = {11},
  pages   = {7094--7112},
  year    = {2018}
}

@article{con2025oblivious,
  title={Oblivious Deletion Codes},
  author={R. Con and R. Li},
  journal={arXiv preprint arXiv:2506.18878},
  year={2025}
}

@article{GGRW23,
  author  = {R. Gabrys and V. Guruswami and J. Ribeiro and K. Wu},
  title   = {Beyond Single-Deletion Correcting Codes: Substitutions and Transpositions},
  journal = {IEEE Transactions on Information Theory},
  volume  = {69},
  number  = {1},
  pages   = {169--186},
  year    = {2023},
  doi     = {10.1109/TIT.2022.3202856}
}

@article{HS21survey,
  author  = {Haeupler, Bernhard and Shahrasbi, Amirbehshad},
  title   = {Synchronization Strings and Codes for Insertions and Deletions---A Survey},
  journal = {IEEE Transactions on Information Theory},
  volume  = {67},
  number  = {6},
  pages   = {3190--3206},
  year    = {2021},
  doi     = {10.1109/TIT.2021.3056317}
}

\end{document}